\documentclass[12pt,a4paper]{amsart}
\usepackage{nicefrac}
\usepackage{enumitem}
\usepackage[pdftex,bookmarks,colorlinks=true,linkcolor=blue,urlcolor=red]{hyperref}
\usepackage{amssymb,amsfonts}
\usepackage{float}
\usepackage{graphicx}
\usepackage{subcaption}
\usepackage{wrapfig}

\numberwithin{equation}{section}
\newcommand{\D}{\displaystyle}

\newcommand{\car}{\mathbf{1}}
\newcommand{\R}{{\mathbb R}}

\newcommand{\Z}{{\mathbb Z}}
\newcommand{\N}{{\mathbb N}}

\newcommand{\esp}{{\mathbb E}}
\newcommand{\pro}{{\mathbb P}}

\newcommand{\Gc}{{\mathcal G}}
\newcommand{\Ec}{{\mathcal E}}

\newcommand{\mb}[1]{\mathbf{#1}}

\newcommand{\tr}{{\rm tr}\,}

\newcommand{\dist}{{\rm dist}\,}

\DeclareMathOperator*{\essinf}{ess\,inf}
\theoremstyle{plain}
\newtheorem{Th}{Theorem}[section]
\newtheorem{Le}{Lemma}[section]
\newtheorem{Pro}{Proposition}[section]
\newtheorem{Cor}{Corollary}[section]
\theoremstyle{definition}
\newtheorem{Rem}{Remark}[section]
\newtheorem*{Rems*}{Remarks}

\title{On the spatial extent of localized eigenfunctions for random Schr\"odinger operators}
\author{Fr{\'e}d{\'e}ric Klopp} \address[Fr{\'e}d{\'e}ric
Klopp]{\vskip.1cm Sorbonne Universit{\'e}, Universit{\'e} de Paris, CNRS,
  Institut de Math{\'e}matiques Jussieu - Paris Rive Gauche, F-75005,
  Paris, France}
\email{\href{mailto:frederic.klopp@imj-prg.fr}{frederic.klopp@imj-prg.fr}}
\author{Jeffrey Schenker}
\address[Jeffrey Schenker]{\vskip.1cm Michigan State University Department of Mathematics, 619 Red Cedar Road, East Lansing MI 48824, USA}
\email{\href{mailto:schenke6@msu.edu}{schenke6@msu.edu}}
\begin{document}
%
\begin{abstract}
  The present paper is devoted to new, improved bounds for the
  eigenfunctions of random operators in the localized regime.  We
  prove that, in the localized regime with good probability, each
  eigenfunction is exponentially decaying outside a ball of a certain
  radius, which we call the ``localization onset length''.  For
  $\ell>0$ large, we count the number of eigenfunctions having onset
  length larger than $\ell$ and find it to be smaller than
  $\exp(-C\ell)$ times the total number of eigenfunctions in the
  system.  Thus, most eigenfunctions localize on finite size balls
  independent of the system size.
\end{abstract}

\maketitle
\section{Introduction}
\label{intro}
Single particle Anderson localization is extremely well studied in
both the physics and mathematics literature.  The phenomenon was first
identified by Anderson in 1958, who argued that for tight-binding
models with sufficiently strong on-site disorder the eigenstates
corresponding to a spectral band may be exponentially localized in
space \cite{PhysRev.109.1492}.  As was later clarified in the
mathematics
literature~\cite{MR0470515,Fr-Sp:83,Fr-Ma-Sc-Sp:85,Si-Wo:86}, this can
be understood as almost-sure pure-point spectrum for the corresponding
Hamiltonians along with exponential decay of the corresponding
eigenfunctions.

Extensive random systems, with sufficient decorrelation between
distant regions, are characterized by the dictum ``anything that can
happen does happen, infinitely often.''  In particular, there are
arbitrarily large regions exhibiting atypical local behavior.  As a
consequence, among the eigenfunctions of a random Hamiltonian
exhibiting Anderson localization, one expects to find examples which
fail to decay over arbitrarily large scales.  For instance an
eigenfunction may have a majority of its mass in a region where the
Hamiltonian is extremely close to a periodic operator over a box of
size $\ell$.  Over this box, the eigenfunction will be close to an
eigenfunction of the corresponding periodic operator, which is
extended over the whole box.  Although the eigenfunction will
eventually decay exponentially, we may need to look far from its
center of localization to observe that decay.

In \cite{MR97m:47002}, del Rio, Jitomirskaya, Last and Simon observed
that dynamical localization is related to spectral localization
through a quantitative bound allowing for the rare a-typical behavior
described in the previous paragraph.  They noted that, for various
models for which localization had been proved, the eigenfunctions were
shown to have the following strong property, which they dubbed
\emph{Semi-Uniformly Localized Eigenfunctions} or SULE.
An operator has a SULE basis if there is a basis $\varphi_j,$
$j=1,\ldots,\infty$, of eigenfunctions such that for some $\mu >0$ it
holds that
\begin{equation}
  \label{eq:SULE0}
  |\varphi_j(x)| \ \le \ C_{\epsilon} e^{\epsilon |x_j|} e^{-\mu|x-x^{(j)}|}
\end{equation}
where $x^{(j)}$ is a \emph{localization center} for $\psi_j$ and
$C_{\epsilon}$ is finite for any $\epsilon>0$.  In the random context,
the decay constant $\mu$ is assumed not to depend on the realization
of the disorder, but the constants $C_\epsilon$ may be disorder
dependent (although finite almost surely).  In subsequent work, it was
observed that in many cases one may obtain the sharper bound
\begin{equation}
  \label{eq:polySULE}
  |\varphi_j(x)| \ \le \ C_{\nu} (1+ |x^{(j)}|)^\nu e^{-\mu|x-x^{(j)}|}
\end{equation}
for $\nu > \nicefrac{d}{2}$, with $C_{\nu}$ finite almost surely.

The exponentially growing prefactor in \eqref{eq:SULE0} is relevant
only if
$|x-x^{(j)}| > \frac{\epsilon}{\mu}|x^{(j)}| + \log
C_{\omega;\epsilon}$; for $x$ closer to $x^{(j)}$, we may replace
\eqref{eq:SULE0} by the simple bound $|\psi_j(x)|\le 1$ valid for any
$\ell^2$-normalized lattice function. Thus the presence of the
prefactor $e^{\epsilon |x^{(j)}|}$ accounts for the possibility, at
large scales, that some eigenfunctions may be extended over a large
region outside of which exponential decay sets in.  However, this
possibility is dealt with coarsely in eq.\ \eqref{eq:SULE0}, since all
eigenfunctions with localization centers far from the origin are
painted with the same brush.  In fact, one expects many of these
eigenfunctions to satisfy a better bound, without the prefactor
$e^{\epsilon |x^{(j)}|}$.

Our main goal here is to present a refinement of SULE that manifests
the fact that most eigenfunctions, wherever they may be localized, are
well localized on a region of size $O(1)$.  Roughly, we accomplish
this by associating to each eigenfunction a ``localization volume'',
which informally is the size of the smallest lattice box outside of
which the eigenfunction exhibits exponential decay. Our main result is
the following improvement of \eqref{eq:SULE0}, which we prove holds
throughout the localization regime of the Anderson model:
\begin{quote}
  there is $\mu >0$ such that, with probability one, the eigenfunctions
  of the Anderson model satisfy
  \begin{equation}
    \label{eq:mainresult} |\varphi_j(x)| \ \leq\|\varphi_j\|_\infty \
    e^{-\mu(|x-x^{(j)}|-\ell_j)_+} ,
  \end{equation}
  where $x^{(j)}\in \Z^d$, $\ell_j\ge 0$, $(a)_+=\nicefrac{(|a|+a)}{2}$
  denotes the positive part and
  $\D\|\varphi_j\|_\infty=\max_{x\in\Z^d}|\varphi_j(x)|$; furthermore,
  there are $C>0$ and $\ell_0>0$, such that, for $\ell>\ell_0$, one
  has
  \begin{equation}
    \label{eq:mainDOSresult}\limsup_{L\rightarrow
      \infty} \frac{\# \left \{ j \ \middle | \ |x^{(j)}|\le L  \quad
        \text{and} \quad \ell_j \geq\ell \right \} }{(2L+1)^d } \ \le \
    e^{-C \ell}.
  \end{equation} 
\end{quote}

The key point here is \eqref{eq:mainDOSresult}, which bounds the
number of eigenfunctions $\psi_j$ for which the length $\ell_j$ is
large.  We emphasize that $\ell_j$ is \emph{not} the ``localization
length" of $\phi_j$.  Indeed, the localization length is a disorder
independent function of the energy of the eigenfunction.  It is the
scale over which exponential decay occurs in the tail of the
eigenfunction, and is bounded by $\nicefrac{1}{\mu}$ for \emph{all}
functions in a SULE basis.  Rather, $\ell_j$ is the length scale at
which localization phenomenon sets in for the particular eigenfunction
$\phi_j$.  For this reason, we refer to $\ell_j$ as the
\emph{localization onset length}, or \emph{onset length}, of $\phi_j$.
As will be clear from the proof, eigenfunctions $\phi_j$ with large
onset lengths $\ell_j$ are associate to rare behavior over the region
$|x-x_E|\le \ell_j$, which can be controlled by large deviation
estimates.  Thus we expect the local behavior observed at scales below
the onset length to be stochastic and highly dependent on the local
environment.

Note that eq.\ \eqref{eq:mainresult} is useful only if
$|x-x_j| \ > \ \ell_j ,$ since for smaller $|x-x_j|$ the bound
saturates to become $|\varphi_j(x)|\le \|\varphi_j\|_\infty$. Thus
eq.\ \eqref{eq:mainresult} implies that the localization volume of
$\varphi_j$ is smaller than $C \ell_j^d$, and the quantity on the
left-hand side of eq.\ \eqref{eq:mainDOSresult} is roughly the
``density of states with localization volume larger than
$C \ell^{d}$.'' Bu eq.\ \eqref{eq:mainDOSresult}, the density of
states with localization volume larger than $V$ is therefore bounded
above by $ \exp (- C V^{\nicefrac{1}{d}}).$ Below we formulate the
above notions in a way that is local in energy, allowing for operators
that have localization only in part of their spectrum.

Finally, let us note that the present result provides a structural
description of the eigenfunction of a random system in the
localization phase that is akin (though quite different for obvious
reasons) to the one found for one dimensional quasi-periodic models
in~\cite{MR3779957}.

\subsection{Outline of the paper}
We formulate precise statements of our main results for the discrete
Anderson model (Thms. \ref{thr:1} and \ref{thr:3}) and their extension
to continuum Schr{\"o}dinger operators and more general tight-binding
models (Thm. \ref{thr:8}) in \S\ref{sec:form-stat-results}.  A brief
overview of the proof of Thm. \ref{thr:1} is given in
\S\ref{sec:outline}.  In \S\ref{sec:numerics} we present the results
of numerical calculation of the onset length for the 1D Anderson model
on finite intervals.  The proofs of Thm.\ \ref{thr:1} and Thm.\
\ref{thr:8} are in \S\ref{sec:proof-theorem} and the proof of
Thm. \ref{thr:3} is in \S\ref{sec:proof-theorem-3}.  In two appendices
we present: 1) a derivation of the SULE estimate from known bounds on
eigenfunction correlators (\S\ref{sec:SULE}) and 2) a large deviation
estimate that is at the heart of the proof of Thm.\ \ref{thr:1} (\S
\ref{sec:LDP}).

\section{Formal statement of results}
\label{sec:form-stat-results}
\subsection{Assumptions}
We focus on the lattice Anderson model, the Hamiltonian
$H_\omega=-\Delta+V_\omega$ on $\ell^2(\Z^d)$, where
\begin{equation}
  \label{eq:11}
  -\Delta \psi(x) \ = \ \sum_{|x-y|=1} \psi(y)
\end{equation}
is the discrete Laplacian and
\begin{equation*}
  V_\omega\psi(x) \ = \ \omega_x \psi(x)  
\end{equation*}
with $\omega=(\omega_\gamma)_{\gamma\in\Z^d}$ a collection of i.i.d. random
variables.  We formulate our results in terms of the restrictions of
$H_\omega$ to regions $\Omega \subset \Z^d$.  For simplicity we take Dirichlet
boundary conditions
\begin{equation*}
  (H_\omega)_{\Omega} \ = \ I_{\Omega}^TH_\omega I_{\Omega }
\end{equation*}
where $I_{\Omega }$ denotes the injection
$I_{\Omega}:\ell^2(\Omega)\rightarrow \ell^2(\Z^d).$ Boundary
conditions play little role in our analysis; the proofs given below
could easily be adapted to other standard conditions, e.g., Neumann or
periodic.

Given a region $\Omega \subset \Z^d$, let
\begin{equation}\label{eq:eigenvalues}
  \Ec((H_\omega)_{\Omega}) \ := \ \left \{ \text{the set of eigenvalues of $(H_\omega)_{\Omega}$} \right \} \ .
\end{equation}
Of course, for finite $\Omega$ we have
$\Ec((H_\omega)_{\Omega})=\sigma((H_\omega)_\Omega)$, the spectrum of
$(H_\omega)_\Omega$.  For infinite $\Omega$, $\Ec((H_\omega)_\Omega)$
is a countable dense subset of the point spectrum
$\sigma_p((H_\omega)_\Omega)$.  It is well known (see
e.g.~\cite[Chapter 3]{MR3364516}) that,
\begin{enumerate}
\item[(A1)] There exist closed sets $\Sigma_p\subset\Sigma\subset\R$
  such the spectrum $\sigma(H_\omega)=\Sigma$ almost surely and the
  point spectrum $\sigma_p(H_\omega)=\Sigma_p$ almost surely.
\item[(A2)] For any $E\in\R$, $\omega$-a.s., the integrated density of
  states
  \begin{equation}\label{eq:DOS}
    N(E) \ :=\ \lim_{\substack{\Omega \uparrow \Z^d \\ \Omega \text{ finite}}} \, \frac{\#\Ec((H_\omega)_{\Omega})\cap(-\infty,E] }{\#\Omega}
  \end{equation}
  exists and is almost surely independent of $\omega$; it is the
  cumulative distribution function of a probability measure supported
  on $\Sigma$.
\end{enumerate}
\begin{Rem}
  The notation $\lim_{\Omega \uparrow \Z^d}$ in \eqref{eq:DOS} denotes
  convergence along the net of finite subsets partially ordered by
  inclusion.  To compute the limit, it suffices to take a sequence of
  centered cubes.  We denote the density of states measure also by
  $N$, taking $N(I)=\int_I dN(E)$.
\end{Rem}

We require a version of the SULE estimates for the random operators
under consideration.  These are conveniently expressed in terms of
weighted norms of eigenfunctions with ``localization center'' in a
given region.  To make the notion of ``localization center'' precise,
to any function $\varphi\in \ell^2(\Omega)$, with $\Omega\subset \Z^d$, we
associate the \emph{set of localization centers}:
\begin{equation}
  \label{eq:loccenters}
  \mathcal{C}(\varphi) \ := \ \left \{ x\in \Omega \ : \ |\varphi(x)| = \|\varphi\|_\infty \right \}.
\end{equation}
Note that this is a non-empty, finite set for any
$\varphi\in \ell^2(\Omega)$.  We have:
\begin{enumerate}
\item[(A3)] There exists $I_{\mathrm{AL}}\subseteq\Sigma$, a union of
  finitely many open intervals of positive length, such that for any
  region $\Omega$, with probability one $(H_\omega)_\Omega$ has pure
  point spectrum on $I_{\mathrm{AL}}$. Furthermore, there are
  constants $A_{\mathrm{AL}}<\infty$, $\mu >0$ such that if $\Omega\subset \Z^d$
  is a region, $S \subset \Omega$ is a finite set, and $0<\varepsilon<1$,
  then, with probability larger than $1-\varepsilon$, any
  $\ell^2$-normalized eigenfunction $\varphi_E$ of
  $(H_\omega)_{\Omega}$ with eigenvalue
  $E\in I_{\mathrm{AL}}\cap \Ec ((H_\omega)_{\Omega})$ and
  $\mathcal{C}(\varphi_E)\cap S\neq \emptyset$ satisfies
  \begin{equation}
    \label{eq:SULE}
    \max_{y\in \mathcal{C}(\varphi_E)\cap S} \left (\sum_{x\in \Omega}e^{2 \mu |x-y|} |\varphi_E(x)|^2 \right )^{\frac{1}{2}}\ \leq\ A_{\mathrm{AL}} \,\left (\tfrac{\# S}{\varepsilon}\right )^{\nicefrac{1}{2}} \ .
  \end{equation} 
\end{enumerate}
\begin{Rem}
  \label{rem:3} 1) This result follows from well known estimates on
  eigenfunction correlators (see \cite[Chapter 7]{MR3364516}). For
  completeness we recall the proof in the appendix. 2) The
  $\ell^2$-estimate \eqref{eq:SULE} directly implies the pointwise
  bound
  $|\varphi_E(x)| \le A_{\mathrm{AL}} \bigl ( \frac{\# S}{\varepsilon}
  \bigr )^{\nicefrac{1}{2}} e^{-\mu|x-x_E|}$ for
  $x_E\in \mathcal{C}(\varphi_E)\cap S$. Since the statement is restricted to
  eigenfunctions with centers in a finite region $S$, we have the
  uniform bound
  $\bigl( \frac{\# S}{\varepsilon} \bigr )^{\nicefrac{1}{2}}$ in place
  of the growing prefactor $(1+|x|)^\nu$ in \eqref{eq:polySULE}. 3) We
  use the \emph{max-norm} on $\Z^d$:
  $$|x| \ = \ \max_{m=1,\ldots,d} |x_m|.$$ 
\end{Rem}

Finally, we need the well known Minami estimate for the Anderson model
(see \cite[Chapter 17]{MR3364516}):
\begin{enumerate}
\item[(A4)] There is a constant $A_M >0$, such that for any finite
  region $\Omega \subset \Z^d$, one has
  \begin{equation}
    \label{eq:40}
    \sup_{E\in I_{\mathrm{AL}}}\pro\left(\left\{\tr(\car_{[E-\varepsilon,E+\varepsilon]}
        ((H_\omega)_{\Omega}))\geq2\right\}\right) \  \le \
    A_M \, |\Omega|^2 \varepsilon^2 \ 
  \end{equation}
  for any $\epsilon>0$.
\end{enumerate}
\begin{Rem} A weaker version would be sufficient to derive our results
  (see \S\ref{sec:proof-theorem} and \S\ref{sec:general-setting}).
\end{Rem}

\subsection{Main Results} 
We begin by reformulating (A3) in terms of the onset lengths for the
eigenfunctions, based on some general notions for functions in
$\ell^2(\Omega)$. Given $\mu>0$, define a family of weighted $\ell^2$
norms by
\begin{equation}\label{eq:Mellphi}
  M_\ell^\mu (\varphi;y) \ := \ \left ( \sum_{x\in \Omega} e^{2\mu(|x-y|-\ell)_+} |\varphi(x)|^2 \right )^{\frac{1}{2}},
\end{equation}
where $\ell=0,1,2,\ldots$. If $M^\mu_\ell(\varphi;y)$ is finite, then
$\D\lim_{\ell\rightarrow
  \infty}M^\mu_\ell(\varphi;y)=\|\varphi\|_{\ell^2}$, by dominated
convergence.  Thus it makes sense to define
$$\ell_\mu(\varphi;y) \ := \ \min \{ \ell \ : \ M_\ell^\mu(\varphi;y) \le 2\|\varphi\|_2\},$$
where the threshold $2$ could be replaced by any fixed number $>1$.

Applying these notions to the SULE estimate in Assumption (A3), we
obtain the following
\begin{Pro}\label{proposition} Let $I_{\mathrm{AL}}$, $\mu$, and
  $A_{\mathrm{AL}}$ be as in (A3) and let $\Omega\subset \Z^d$ be region and
  $S\subset \Omega$ a finite set.  If $0<\varepsilon<1$, then, with
  probability larger than $1-\varepsilon$, any $\ell^2$-normalized
  eigenfunction $\varphi_E$ of $(H_\omega)_{\Omega}$ with eigenvalue
  $E\in I_{\mathrm{AL}}\cap\Ec((H_\omega)_{\Omega})$ and
  $\mathcal{C}(\varphi_E)\cap S\neq \emptyset$ satisfies
  \begin{equation}\label{eq:onsetSULE}
    \left (\sum_{x\in \Omega}e^{2 \mu (|x-x_E|-\ell_E)_+ } |\varphi_E(x)|^2 \right )^{\frac{1}{2}}\ \leq\  \ 2 ,
  \end{equation} 
  for any $x_E\in \mathcal{C}(\varphi_E)\cap S$ with
  $\ell_E = \ell_\mu(\varphi_E;x_E) < \tfrac{1}{\mu}\left (\log
    A_{\mathrm{AL}} + \tfrac{1}{2} \log \#S - \tfrac{1}{2} \log
    3\varepsilon \right ) +1$.
\end{Pro}
\begin{proof}
  Note that
  \begin{equation*}
    \sum_{x\in \Omega} e^{2\mu(|x-x_E|-\ell)_+} |\varphi_E(x)|^2  \ \le \ \sum_{|x-x_E|\le \ell}  |\varphi_E(x)|^2   + e^{-2\mu \ell }\sum_{|x-x_E|> \ell}  e^{2\mu|x-x_E|} |\varphi_E(x)|^2  \
    \le 1 + e^{-2\mu \ell} A_{\mathrm{AL}}^2 \,\tfrac{\#
      S}{\varepsilon}
  \end{equation*}
  Thus $M_\ell^\mu(\varphi_E;x_E) \le 2 $ provided
  $ 2 \log A_{\mathrm{AL}} + \log \#S -\log \varepsilon -2 \mu \ell  \le \log 3$.  The
  result follows from the definition of $\ell_\mu(\varphi_E;x_E)$.
\end{proof}
 
Although an eigenfunction $\varphi_E$ may have more than one
localization center, given two localization centers
$x_E,x'_E\in\mathcal{C}(\varphi_E)$ one has
\begin{equation}
  \label{eq:9}
  |\ell_\mu (\varphi_E;x_E)-\ell_\mu (\varphi_E;x'_E)|\leq |x_E-x'_E|
\end{equation}
This follows immediately from the definition of the onset length and
the bound
\begin{equation*}
  \sum_{x\in \Omega} e^{2\mu(|x-x_E|-(\ell_\mu (\varphi_E;x'_E)+|x_E-x'_E|))_+} |\varphi_E(x)|^2 \ \leq \
  \sum_{x\in \Omega} e^{2\mu(|x-x'_E|-\ell_\mu (\varphi_E;x'_E))_+}
  |\varphi_E(x)|^2 \ \leq \ 4.
\end{equation*}
We note that the onset length $\ell_\mu (\phi_E;x_E)$ also gives an upper
bound on the diameter of $\mathcal{C}(\varphi_E)$, namely,
\begin{Pro}
  \label{pro:1}
  Pick $\kappa>0$ such that $8e^{-\mu \kappa}=1$. If
  $(x_E,x'_E)\in\mathcal{C}(\varphi_E)^2$ then
  \begin{equation}
    \label{eq:8}
    |x_E-x'_E|\leq\ell_\mu (\varphi_E;x_E)+
    \frac d{2\mu }\log\left(2\ell_\mu (\varphi_E;x_E)+2\kappa+1\right)
    +\frac{3 \log2}{2\mu }
  \end{equation}
\end{Pro}
\begin{proof}
  By the definition of $\ell_\mu (\varphi_E;x_E)$, if
  $(x_E,x'_E)\in\mathcal{C}(\varphi_E)^2$ with
  $|x_E-x'_E|\geq \ell_\mu (\varphi_E;x_E)$, then
  \begin{equation*}
    \frac1{(2\ell_\mu (\varphi_E;x_E)+2\kappa+1)^{d/2}}
    \ \leq\ \sqrt{2} \|\varphi_E\|_\infty \ = \ \sqrt{2} |\varphi_E(x'_E)| \
    \leq \ 2^{\nicefrac{3}{2}} e^{-\mu (|x'_E-x_E|-\ell_\mu (\varphi_E;x_E))} \ ,
  \end{equation*}
  with the first inequality coming from Lemma~\ref{le:3} below. Taking
  the logarithm, we get~\eqref{eq:8}.
\end{proof}

The bounds~\eqref{eq:9} and~\eqref{eq:8} show that, when the
localization center is not unique, the onset lengths for two distinct
localization centers differ at most by a universal constant factor.
Moreover, one obtains the pointwise bound
\begin{equation}
  \label{eq:13}
  |\varphi_E(x)|\ \leq \ \|\varphi_E\|_\infty
  e^{-\mu (|x-x_E|-\tilde{\ell}_{\mu,E})_+}\ , \quad \text{ where
  }\tilde{\ell}_{\mu,E} =\ell_\mu (\varphi_E;x_E)+\frac{d}2
  \log(2\ell_\mu (\varphi_E;x_E)+2\kappa+1) \ ,
\end{equation}
corresponding to eq.\ \eqref{eq:mainresult}.

Our main result quantifies the distribution of onset lengths for
eigenfunctions with localization centers in a given bounded region.
It is convenient to take these bounded regions to be lattice cubes:
\begin{equation}
  \Lambda_L(x_0) \ = \ \left (x_0 + \left
      ]-\tfrac{L}{2},\tfrac{L}{2} \right ]^d \right )\cap \Z^d \ ,
\end{equation}
for $L=1,2,3,\ldots$.  We note that these cubes have diameter
$\operatorname{diam} \Lambda_L(x_0)=L$ (in the max-norm metric) and
volume $\# \Lambda_L(x_0)=L^d$.  We call $x_0$ the \emph{center} of
the cube $\Lambda_L(x_0)$; for odd $L$ it is the geometric center, but
for even $L$ it is integer part of the geometric center.  Our main
result quantifies the eigenvalues associated to eigenfunctions with
onset length larger than a given number.
\begin{Th}
  \label{thr:1}
  Let $\mu$ and $I_{\mathrm{AL}}$ be as in (A3) and let
  $[a,b]\subset I_{\mathrm{AL}}$.  Then, for any $\nu <\mu$ and any $p>0$
  there exist $\ell_0>0$, and $L_{0}>0$ such that if
  $\Omega\subset\Z^d$ is a region with
  $\Lambda = \Lambda_L(x_0)\subset \Omega$ with $L\ge L_0$, then with
  probability larger than $1-L^{-p}$, for all $\ell\ge \ell_0$, one
  has
  \begin{multline}
    \label{eq:4}
    \# \bigl \{ E \in \Ec((H_\omega )_{\Omega})\cap [a,b] \ : \
    \mathcal{C}(\varphi_E)\cap \Lambda\neq \emptyset \text{ and }
    \ell_\nu(\varphi_E;x_E) \ge \ell \text{ for
    }x_E\in \mathcal{C}(\varphi_E)\cap \Lambda \bigr \} \\ \leq \ L^d e^{-C_\nu \ell }
  \end{multline}
  where $\D C_\nu=\frac13 \min\left(1,\frac{\mu -\nu}{\nu} \right)$.
\end{Th}
\noindent Theorem \ref{thr:1} is proved in \S\ref{sec:proof-theorem}
below. As an immediate consequence of this result and the ergodic
properties of the Hamiltonian $H_\omega$, we have the following:
\begin{Cor}
  \label{cor:2}
  For $\nu <\mu$, $a<b$ real such that $[a,b]\subset I_{\mathrm{AL}}$, and
  $\ell>0$, the limit
  \begin{multline*}
    N_\nu([a,b],\ell) \ := \\ \lim_{L\to+\infty} \frac{\#\{ E
      \in \Ec(H_\omega )\cap [a,b] \ : \ \mathcal{C}(\varphi_E)\cap \Lambda\neq
      \emptyset \text{ and } \ell_\nu(\varphi_E;x_E) \ge \ell \text{ for
        some }x_E\in \mathcal{C}(\varphi_E)\cap \Lambda \}}{N(I)\cdot L^d }
  \end{multline*}
  exists almost surely and is a.s.\ independent of $\omega$.  Moreover
  there is a Borel probability measure $P_\nu$ on $I\times \N $ such that
  $$ N_\nu([a,b],\ell) \ = \ P_\nu ([a,b]\times [\ell,\infty)),$$
  and there exists $\ell_0>0$ such that, for $\ell\geq \ell_0$ and
  $a<b$ real such that $[a,b]\subset I$
  \begin{equation}
    \label{eq:3}
    N_\nu([a,b],\ell)\leq
    \frac{N([a,b])}{N(I)}e^{-C_\nu\ell}
  \end{equation}
\end{Cor}
\begin{Rem} Here we take $\Omega=\Z^d$ and recall that $N$ is the
  integrated density of states of $H_\omega$.
\end{Rem}

Let us now turn to the question of finding a lower bound for the left
hand side of \eqref{eq:4}.
To find such a bound, we must construct sufficiently many states with
large onset length. Recalling the classical heuristics of Lifshits
tails, the states that immediately spring to mind are those located
near the edges of the almost sure spectrum. It is well known that the
parts of the spectrum close to its boundary, in particular to the
infimum of the spectrum, belong to the localization region
$I_{\mathrm{AL}}$. We have the following
\begin{Th}
  \label{thr:3}
  Let $\mu$ and $I_{\mathrm{AL}}$ be as in (A3). Let $E_-$ be the
  infimum of $\Sigma$ the almost sure spectrum of $H_\omega$ and
  assume $E_->-\infty$.  Then, there exist $\ell_0>0$ and $L_0>0$ such
  that, for any $\nu <\mu$, for $\Lambda = \Lambda_L$ with $L\ge L_0$,
  and all $\ell\geq \ell_0$, with probability 1, one has
  \begin{multline}
    \label{eq:19}
    \#\{ E \in \Ec(H_\omega) \ : \mathcal{C}(\varphi_E)\cap \Lambda\neq
    \emptyset \text{ and } \ell_\nu(\varphi_E;x_E) \ge \ell \text{ for
      some }x_E\in \mathcal{C}(\varphi_E)\cap \Lambda \} \ \\\geq\ \#\{ E
    \in \Ec(H_\omega)\cap [E_-,E_-+c\ell^{-d-1}]\ : \
    \mathcal{C}(\varphi_E)\cap \Lambda\neq \emptyset \}
  \end{multline}
  where $c$ can be taken such that
  $5(4d)^{\frac2d}c^{\frac{2}{d+1}}=1$.
\end{Th}
\begin{Rem}
  Here we take $\Omega=\Z^d$.
\end{Rem}

The proof of Theorem~\ref{thr:3} can be found in
\S\ref{sec:proof-theorem-3}.  As will be clear from the proof, the
estimate \eqref{eq:19} is \emph{deterministic}.  Estimating the right
hand side of~\eqref{eq:19}, that is the number of eigenvalues of our
random operator inside $[E_-,E_-+ct^{-d-1}]$ having at least one
localization center in $\Lambda$ in terms of the volume of $\Lambda$
and the integrated density of states, will yield a random
estimate. Such bounds are akin to Lifshitz tail estimates for which it
is usually the operator that is restricted to a finite volume rather
than the localization centers (both approaches are equivalent in the
localization regime; see e.g.~\cite{MR3273314}).

One also has the corresponding infinite volume estimate, namely,
\begin{Cor}
  \label{cor:1}
  Let $E_-$ be the infimum of $\Sigma$ the almost sure spectrum of
  $H_\omega$ and assume $E_->-\infty$.  Then there exists $\ell_0>0$
  such that, for any $\nu <\mu$ and $\ell\geq \ell_0$, one has
  \begin{equation}
    \label{eq:1}
    N_\nu(\Sigma,\ell)\geq N(E_-+c\ell^{-d-1})
  \end{equation}
  where $c$ is taken as in Theorem~\ref{thr:3}.
\end{Cor}
%
The asymptotic behavior of the integrated density of states $N$ near
$E_-$ is a classical topic of study of random media and is known for
many models. For example, for the Anderson model it is well known
that, for $\lambda>0$ small,
$N(E_-+\lambda)\geq e^{-f(\lambda)\lambda^{-d/2}}$ where
$f:[0,\infty|\to[0,+\infty[$ is a decreasing function that depends on
the tail of the common distribution of the random variables
$(\omega_x)_{x\in\Z^d}$ near their almost sure minimum; in particular,
$f$ diverges at most logarithmically at $0$ if this tail does not fall
off faster than polynomially (see e.g.~\cite{MR3364516,MR2307751}).
Using this lower bound, in dimension 1, the bound~\eqref{eq:1} becomes
\begin{equation*}
  \#\{ E \in \Ec(H_\omega) \ : \ 
  \mathcal{C}(\varphi_E)\cap \Lambda\neq \emptyset \text{ and }
  \ell_\nu(\varphi_E;x_E) \ge \ell \text{ for some
  }x_E\in \mathcal{C}(\varphi_E)\cap \Lambda \} \ \geq\ L^d e^{-f(\ell^{-2})\ell} \ .
\end{equation*}
In particular, we see that in dimension $1$ the upper bound
\eqref{eq:4} is matched by a lower bound of the same magnitude, up to
the prefactor $f(\ell^{-2})$ that is of lower order.
Nevertheless, the lower bound given only by the ``Lifshitz tail
states'' should not be optimal, as these eigenvalues live
energetically in very tiny regions at the edges of the spectrum. It
seems reasonable to expect the upper bound~\eqref{eq:3} to be optimal.
\subsection{A more general setting}
\label{sec:general-setting}
The results of the previous section can be extended in a
straightforward way to more general random Schr{\"o}dinger operators. We
turn to this now. To avoid certain technicalities, we only consider
random Schr{\"o}dinger operators $H_\omega=-\Delta+V_\omega$ on $\R^d$ or
$\Z^d$ that are $\Z^d$-ergodic and such that the sup-norm of
$V_\omega$ is almost surely bounded by a fixed finite constant. In
particular, there exists a closed set $\Sigma \subset \R$, bounded from below
for operators on $\R^d$ and bounded for operators on $\Z^d$, such that
$\sigma(H_\omega)=\Sigma$ almost surely.  To avoid unnecessary
complications due to possible singularities, we will not give
pointwise bounds for the eigenfunction of operators on the continuum,
but rather bounds on local $L^2$-norms. Therefore, for $x\in\Z^d$, we
set
\begin{itemize}
\item
  $\D\|\varphi\|_2(x)=\|\varphi\|_{L^2(x+]-\nicefrac{1}{2},\nicefrac{1}{2}]^d)}$
  in the case of an operator on $\R^d$, for $\varphi\in L^2(\R^d)$.
\item $\|\varphi\|_2(x)=|\varphi(x)|$ in the case of an operator on
  $\Z^d$, for $\varphi\in \ell^2(\Z^d)$,
\end{itemize}
For the sake of simplicity, we also restrict ourselves to the region
$\Omega=\R^d$ or $\Z^d$ (see section~\ref{sec:form-stat-results})
and the finite regions $S$ we deal with will only be cubes that are
centered at points of $\Z^d$ and have integer side length. Depending
on the context, they will be cubes in $\R^d$ or their restrictions to
$\Z^d$. As above, $(H_\omega)_{\Lambda}$ denotes the restriction of
$H_\omega$ to $\Lambda$ with Dirichlet boundary conditions. As will
follow from the proofs, by their very nature (i.e. the use of
localization), the arguments are valid for other self-adjoint boundary
conditions.

As before, we define the set of localization centers of a normalized
square integrable function $\varphi$ on $\Lambda$ (where $\Lambda$ is
a cube centered at a point in $\Z^d$ having integer or infinite side
length)
\begin{equation}
  \label{eq:10}
  \mathcal{C}(\varphi) \ := \ \left \{ x\in \Lambda\ : \ \|\varphi\|(x) =
    \|\varphi\|_{2,\infty }\right \}\quad\text{where}\quad
  \|\varphi\|_{2,\infty }=\max_{x\in\Lambda}\|\varphi\|_2(x). 
\end{equation}
Our assumptions are:
\begin{enumerate}
\item[(IAD)] There exists $r>0$ such that, if $\Lambda$ and $\Lambda'$
  are cubes such that $d(\Lambda,\Lambda')>r$ then, the finite volume
  operators $(H_\omega)_{\Lambda}$ and $(H_\omega)_{\Lambda'}$ are
  stochastically independent.
\item[(Loc)] There exists a compact non empty interval
  $I\subset\Sigma$ such that $H_\omega$ has pure point spectrum on $I$
  almost surely; and there exists positive real numbers $\xi\in(0,1)$,
  $p>0$, $q>0$, $L_{\mathrm{fin}}>0$ such that for
  $L\geq L_{\mathrm{fin}}$, with probability at least $1-L^{-p}$, for
  every eigenvalue $E\in I\cap\Ec(H_\omega)$ with associated normalized
  eigenvector $\varphi_E$ such that
  $\mathcal{C}(\varphi_E)\cap \Lambda_L \neq \emptyset$, one has
  \begin{equation}
    \label{eq:14}
    \max_{y\in \mathcal{C}(\varphi_E)\cap \Lambda_L} \left
      (\sum_{x\in\Z^d} e^{2|x-y|^\xi}\|\varphi_E\|^2_2(x) \right
    )^{\frac{1}{2}}\leq L ^q
  \end{equation}
\item[(SE)] For $K>1$, there exists $C_K>0$ such that, for
  $\delta\in(0,1]$, one has the following spacing estimate
  \begin{equation}
    \label{eq:15}
    \pro\left\{\exists E\in I;\
      \tr(\car_{[E-\delta,E+\delta]}
      ((H_\omega)_{\Lambda_L}))\geq2\right\} \leq C_K L^{2d}|\log
    \delta|^{-K}.
  \end{equation}
\end{enumerate}
\begin{Rem}
  \label{rem:1}
  1) Assumption (Loc) has been proved for various models in various
  energy regimes (that depend on the model) e.g. the continuous
  Anderson model at the bottom of the spectrum and at internal band
  edges (see e.g.~\cite{MR2998830,MR2207021}), or the displacement
  model at the bottom of the spectrum
  (see~\cite{Klopp_Loss_Nakamura_Stolz_LocalizationRandomDisplacementModel}).
  One could also allow for magnetic fields, etc.
  
  \noindent 2) We chose here to allow for sub-exponential decay in
  place of the exponential decay considered above, as there are
  certain models where, to our knowledge, no better decay estimate has
  been obtained to date (see e.g.~\cite{MR2998830}). As we shall see,
  it will essentially not affect our analysis. In~\eqref{eq:14}, at no
  expense, we could have replaced $H_\omega$ by
  $(H_\omega)_{\Lambda_L}$ and the sum over $\Z^d$ by a sum over
  $\Lambda_L$ (see e.g.~\cite{MR3273314}).
  
  \noindent 3) Except in dimension 1 (see~\cite{MR3200336}), the
  spacing estimate (SE) is known for very few models. For the
  (discrete) Anderson model (and more generally models involving
  independent rank one perturbations), it follows from the Minami
  estimate (\cite{MR97d:82046}) with a better bound: the
  $|\log \delta|^{-K}$ term is replaced with $\delta^2$ as in (A4) above. For
  continuous Anderson models with suitable assumptions on the random
  potentials, it was proved recently in~\cite{MR4228280}.
\end{Rem}

Following the example of section~\ref{sec:form-stat-results}, for
$\varphi_E$ satisfying~\eqref{eq:14} we let
\begin{equation}
  \label{eq:16}
  M_\ell^\xi(\varphi_E;y) \ := \ \left ( \sum_{x\in \Z^d} e^{2(|x-y|-\ell)^\xi_+} \|\varphi\|_2^2(x) \right )^{\frac{1}{2}},
\end{equation}
where $\ell=0,1,2,\ldots$. As previously, we define
\begin{equation*}
  \ell_\xi(\varphi;y) \ := \ \min \{ \ell \ : \ M_\ell^\xi(\varphi;y) \le
  2\|\varphi\|_2\}.
\end{equation*}
Under assumption (Loc), applying these notions to the
estimate~\eqref{eq:14}, we see that there exists $\kappa>0$ such that,
for $L\geq L_{\text{fin}}$, with probability at least $1-L^{-p}$, for
every $E\in I\cap\Ec(H_\omega)$, $\varphi_E$ normalized eigenvector of
$H_\omega$ associated to $E$ and $x_E\in\mathcal{C}(\varphi_E)$, one
has $\ell_\xi(\varphi;x_E)\leq \kappa (\log L)^{1/\xi}$. Moreover, a
straightforward modification of the proof of Proposition~\ref{pro:1}
yields that, for $\kappa>0$ sufficiently large (depending on $\xi$
only), for $(x_E,x'_E)\in\mathcal{C}(\varphi_E)^2$, one has
\begin{equation}
  \label{eq:18}
  |x_E-x'_E|\leq\ell_\mu (\varphi_E;x_E)+
  \kappa\left(\log\left(2\ell_\mu
      (\varphi_E;x_E)+2\kappa+1\right)\right)^{1/\xi}.
\end{equation}

Our main result for the more general models considered here is the
following
\begin{Th}
  \label{thr:8}
  Assume (IAD), (Loc) and (SE). Then, for any $0<\tilde\xi<\xi$, there
  exists $C>0$, $L_{\mathrm{fin}}>0$ and $\ell_0>0$ such that, for any
  $L\geq L_{\mathrm{fin}}$, with probability larger than
  $1-L^{-p}\log L$ (where $p$ is given in assumption (Loc)),
  \begin{itemize}
  \item For $\varphi_E$ any normalized eigenfunction of $H_\omega$
    associated to the eigenvalue $E\in I\cap \Ec(H_\omega)$ such that
    $\mathcal{C}(\varphi_E)\cap \Lambda_L \not=\emptyset$, there exists
    $x_E\in\mathcal{C}(\varphi_E)\cap\Lambda_L$ such that
    \begin{equation}
      \label{eq:43}
      \forall x\in\Z^d,\quad\|\varphi_E\|_2(x)\leq \|\varphi_E
      \|_{2,\infty}e^{-(|x-x_E|-\tilde{\ell}_E)^{\tilde{\xi}}}
    \end{equation}
    where
    $\D\tilde{\ell}_E=\ell_{\xi'}(\varphi_E;x_E)
    +C\max(\log\ell_{\xi'}(\varphi_E;x_E),1)^{1/\xi'}$;
  \item moreover, for $\ell\geq \ell_0$, one has
    \begin{equation}
      \label{eq:17}
      \frac{\#\{ E \in \Ec(H_\omega )\cap I \text{
          associated to }\varphi_E \text{
          s.t. }\exists x_E\in \mathcal{C}(\varphi_E)\cap\Lambda_L\text{ and }\ell_{\xi'}(\varphi_E;x_E)\geq\ell\}}{|\Lambda_L| } \leq 
      e^{-C\ell}.
    \end{equation}
  \end{itemize}
\end{Th}
\noindent Only minor modifications of the proof of Theorem~\ref{thr:1}
yield Theorem~\ref{thr:8}. We state the necessary modifications in
\S~\ref{sec:thr8} below.\\
Our assumptions guarantee the existence of a density of states; hence,
one also recovers the analogue of Corollary~\ref{cor:2} in this
setting. As for lower bounds, the proof of Theorem~\ref{thr:3} and its
corollary~\ref{cor:1} was based on the fact that low lying states have
large onset length; this is still correct in the more general model
under certain assumptions. For example, if $H_\omega=-\Delta+V_\omega$
where $V_\omega$ is an alloy type potential that is almost surely
lower bounded, then the scheme of proof of Theorem~\ref{thr:3} also
works and, \emph{mutatis mutandi}, one gets the same result.
\subsection{Outline of the proof}
\label{sec:outline}
The basic technical lemma leading to the proof of Thm.\ \ref{thr:1}
goes as follows: pick two length scales $L>\ell$; if one knows that
\begin{enumerate}
\item the operator $(H_\omega)_{\Lambda}$ exhibits localization at an
  eigenvalue $E$ in a cube $\Lambda$ of side length $L$ (in the sense
  that the weighted sum in the left hand side of~\eqref{eq:SULE} (for
  $S=\Lambda$) is bounded by $e^{\mu  \ell}$), and
\item if $\varphi_E$, the associated eigenfunction, has a localization
  center in $Q$, a cube of side length $\ell$, such that, when
  enlarging $Q$ somewhat into $\tilde{Q}$, $(H_\omega)_{\tilde{Q}}$
  has at most one eigenvalue at distance $e^{-\mu  \ell}$ to $E$,
\end{enumerate}
then, $(H_\omega)_{\tilde{Q}}$ admits an eigenvalue, say, $\tilde{E}$
exponentially close to $E$ such that the associated eigenvector is
just $\varphi_E$ restricted to $\tilde{Q}$, up to an error of size
$e^{-\mu  \ell/2}$. See Lemma~\ref{le:1} below for a precise statement.

The strategy to obtain Thm.\ \ref{thr:1} from the lemma is the
following.  For a large side length $L$, we define a decreasing finite
sequence of scales (side lengths) $L_n$ by
$L_1 = \lfloor \beta \log L \rfloor$ and
$L_n = \lfloor \tfrac{1}{4} L_{n-1}\rfloor$ for $n=1,\ldots, m=m(L)$
such that $L_m$ is sufficiently large but independent of $L$. For each
generation $n$, we roughly partition our initial cube $\Lambda$ of
side length $L$ into smaller cubes of side length $L_n$. For each
eigenfunction of $H_\omega$ having a localization center in $\Lambda$,
we consider the sequence of the cubes of the different generations
that contain this center of localization. For $n_0\geq1$, we say that
the localization center is \emph{good} from generation 1 to generation
$n_0-1$ if, for $1\leq n \leq n_0-1$, both assumptions (1) and (2) in the basic
technical lemma hold for the cubes of generations $n$ and $n+1$
(i.e. we take $L=L_n$ and $\ell=L_{n+1}$ in assumptions (1) and (2)
above) containing said localization center. Applying the basic
technical lemma inductively to the eigenfunction associated to a good
localization center from generation $1$ to $n_0$, we see that the
associated onset length is at most $L_{n_0}$ and that the associated
eigenfunction decays exponentially outside a cube of generation $n_0$
containing said localization center. Using the independence properties
of the Hamiltonian on the cubes within each generation and estimates
of the probability that either (1) or (2) fail (provided by
assumptions (A3) and (A4)), we can bound the number of localization
centers that fail to be good for generations larger than $n_0$ using a
large deviation principle (see Prop. \ref{prop:LDP}).  We, thus, bound
the number of eigenfunctions having onset length larger than
$L_{n_0}$.
\section{Numerical results}\label{sec:numerics}
In this section we present numerical results for the onset lengths of
eigenfunctions of the $1D$ Anderson model $H_\omega$ with random
potential $\lambda \omega_x$ with $\omega_x$ uniform in the interval
$[-1,1]$.  The spectrum of $H_\omega$ on the full line is the interval
$[-2-\lambda,2+\lambda]$.  As $H_\omega$ is a $1D$ Schr{\"o}dinger
operator, it is well known that localization holds throughout the
spectrum (see, \cite[Chapter 9]{cycon1987} and \cite[Chapter
12]{MR3364516}).  The eigenfunctions of $H_\omega$ decay as
$|x|\rightarrow\infty$ at a rate given by the \emph{Lyapunov
  exponent}, which can be computed using products of transfer
matrices.  Specifically, for each energy $E$ and $n\in \Z$ we define the
\emph{transfer matrix}:
$$T_n(E;\omega) \ := \ \begin{pmatrix}
  \lambda \omega_n-E & 1 \\
  1 & 0
\end{pmatrix} \ . $$ The \emph{Lyapunov exponent} is the limit
\begin{equation}\label{eq:lyap}
  L(E) \ := \ \lim_{n\rightarrow \infty} \frac{1}{n}\log \| T_n(E,\omega)\cdots T_1(E,\omega) \| \ .
\end{equation}
The limit is known to exist and be independent of $\omega$ for almost
every $\omega$ (see \cite[Chapter 9]{cycon1987}).

In Fig.~\ref{Fig:lyapds}, numerical estimates of the Lyapunov exponent
$L(E)$ and density of states $n(E)$ for $H_\omega$ with $\lambda=1$
are shown. As both $L(E)$ and $n(E)$ are symmetric functions of the
energy $E$, these were computed only for $E\ge 0$ (values shown on the
plot for $E<0$ correspond to those computed for $|E|$). The Lyapunov
exponents were estimated at $101$ evenly spaced energy points,
$E_0=0$, $E_1=0.03$, $\ldots$, $E_{100}=3$, by averaging $100$ samples
of $\frac{1}{n}\log \| T_n(E_j,\omega)\cdots T_1(E_j,\omega) \|$ with
$n=10^6$. The density of states was estimated by counting the
proportion of eigenvalues falling in each energy interval
$[E_{j-1},E_j]$, $j=1,\ldots,100$ for the exact diagonalization of
$240$ samples of $(H_\omega)_{\Lambda}$ with $\Lambda=[1,2000]$
($480,000$ total eigenvalues).\footnote{Numerical computations were
  preformed in Matlab on Michigan State University's High Performance
  Computing Center.  To obtain precise computations of eigenfunctions,
  including the exponential tails, we used the open source GEM Library
  \cite{GEM}, which allows for arbitrary precision linear algebra
  computations.}
\begin{figure}
  \centering
  \begin{subfigure}[t]{0.39\textwidth}
    \centering
    \includegraphics[width=\linewidth]{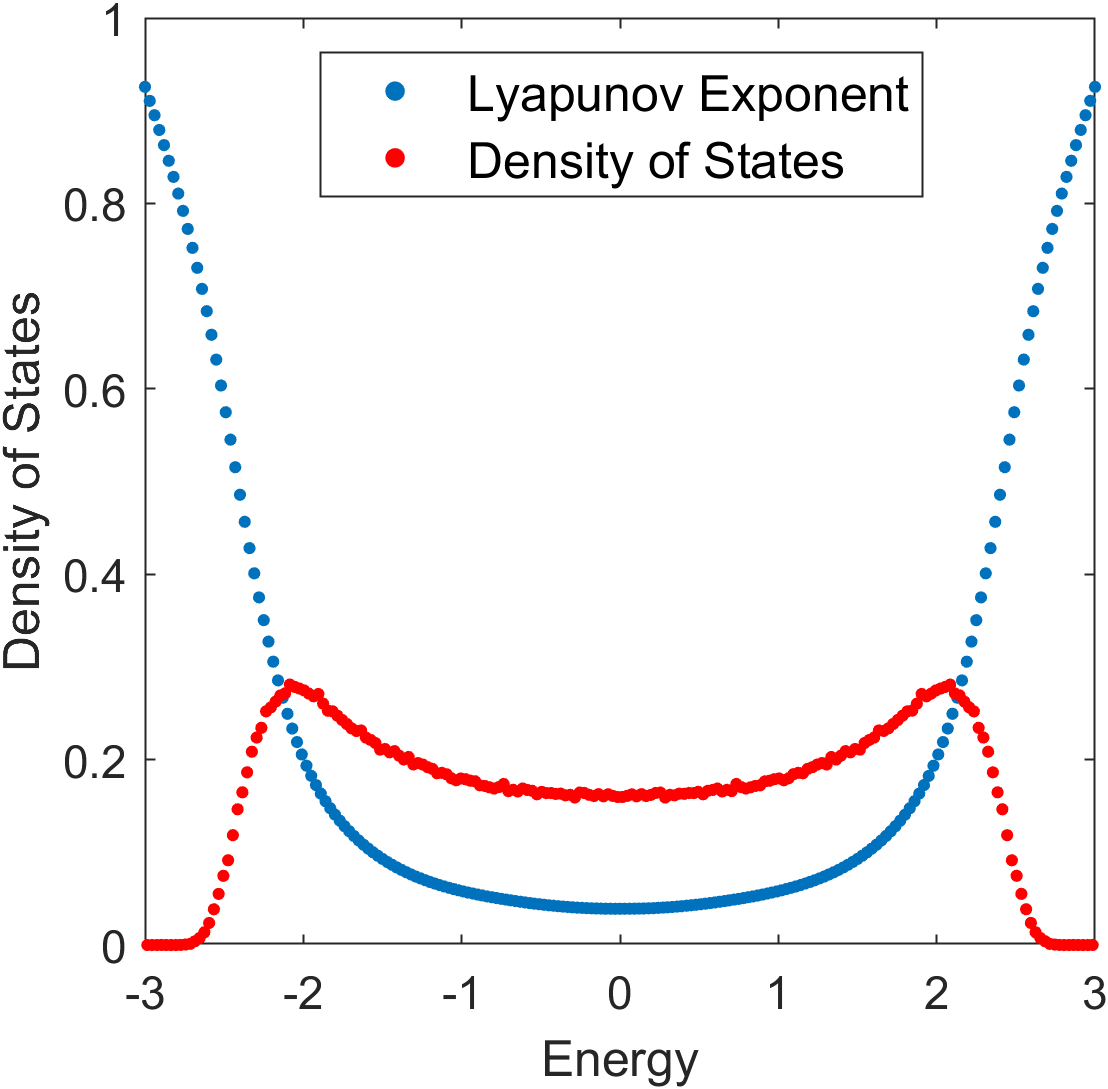}
    \caption{Lyapunov exponent and density of states versus energy.}
    \label{Fig:lyapds}
  \end{subfigure}\hfill
  \begin{subfigure}[t]{0.575\textwidth}
    \centering
    \includegraphics[width=\linewidth]{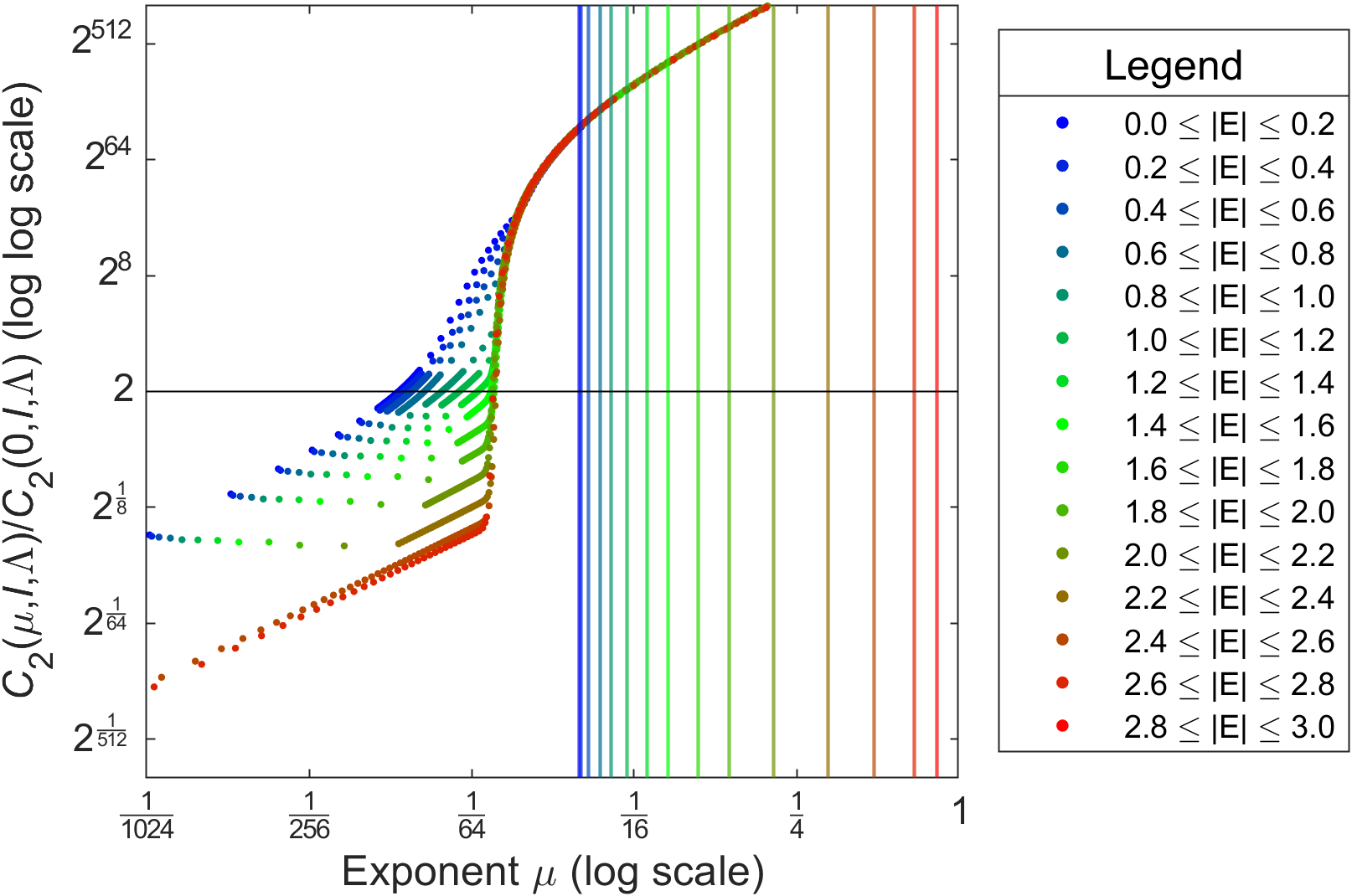} \caption{Normalized
      eigenfunction correlators
      $\nicefrac{C_2(\mu,I,\Lambda)}{C_2(\mu,I,\Lambda)}$ vs.
      exponent $\mu$ for the energy intervals shown. The Lyapunov
      exponent $L(I)$ for each interval is indicated as a vertical
      line.}
    \label{Fig:Correlators}
  \end{subfigure}
  \caption{Lyapunov Exponent, Density of States, and Correlators for
    an interval of length 2000 for the 1D Anderson model with disorder
    $\lambda=1$.}
\end{figure}

As recalled in \S\ref{sec:SULE} below, the localization property (A3)
is a consequence of an exponentially weighted bound on averaged
eigenfunction correlators:
\begin{equation}\label{eq:exploc}
  C_1(\nu,I,\Lambda) \ := \  \frac{1}{\# \Lambda}\sum_{x,y\in \Lambda} e^{\nu |x-y|} \esp \left ( \sum_{E\in I\cap \Ec((H_\omega)_\Lambda)} |\varphi_E(x)| |\varphi_E(y)|  \right )\ \le \ A_\nu  \ ,
\end{equation}
with $\nu >0$ and $A_\nu <\infty$ that may depend on the choice of the
energy interval $I$ but are independent of the region $\Lambda$.
Because the Lyapunov exponent describes the almost certain tail
behavior of eigenfunctions, one is tempted to suppose that
$A_\nu <\infty$ for $\D\nu =
L(I):=\textstyle{\inf_{E\in I}}L(E)$. However, on closer inspection this
seems unlikely.  Indeed, for $\nu = L(I)$ we should have
$\sup_\Lambda C_1(\nu,I,\Lambda)=\infty$, since for any sufficiently
large finite volume $\Lambda$ a certain fraction of the eigenfunctions
will exhibit decay with an exponent smaller than $\nu$, allowing the
exponential weight in \eqref{eq:exploc} to dominate for large
$\Lambda$.  Indeed, known proofs of localization yield a correlator
bound of the form \eqref{eq:exploc} with an exponent $\nu$ that is
strictly less than $L(I)$, e.g., $\nu < \frac{1}{2}L(I)$ in
\cite[Chapter 12]{MR3364516}.  We are not aware of an estimation in
the literature of the exact exponent $\nu_c$ at which $C_1(\nu_c,I)$
diverges, nor of a precise estimate of the divergence as
$\nu \uparrow \nu_c$.

To compute onset lengths, one must choose a particular value of the
exponent $\mu$.  As shown in \S\ref{sec:SULE}, the constant
$A_{\mathrm{AL}}$ appearing in (A3), and thus in the \emph{a priori}
bound on onset lengths in Prop.\ \ref{proposition}, is bounded by the
correlator $C_1(2\mu,I,\Lambda)$.  In a concrete context, it is
important to choose the exponent $\mu$ so that $C_1(2\mu, I,\Lambda)$,
or some other quantity giving an \emph{a priori} bound on the onset
lengths, is not too large.  For the purposes of the numerical
investigations reported here, we found it convenient to work with the
following $\ell^2$, or \emph{density-density}, correlator:
\begin{equation} \label{eq:ell2correlator} C_2(\mu,I,\Lambda) \ = \
  \frac{1}{\#\Lambda}\sum_{x,y\in \Lambda} e^{2\mu |x-y|} \esp \left (
    \sum_{E\in I\cap\Ec((H_\omega)_\Lambda)} |\varphi_E(x)|^2
    |\varphi_E(y)|^2 \right ) \ .
\end{equation}
Since all eigenfunctions are pointwise bounded by $1$, we have the
trivial bound $C_2(\mu,I,\Lambda)\le C_1(2\mu,I,\Lambda)$. In the limit
$\mu \rightarrow 0$, we have
\begin{equation}\label{eq:C20=IDS}
  C_2(0,I,\Lambda) \ = \ \frac{\esp \left (\# (I\cap \Ec((H_\omega)_\Lambda)) \right )}{\# \Lambda} \ ,
\end{equation} 
which is the finite volume integrated density of states on $I$.

\begin{figure}
  \centering
  \begin{subfigure}[t]{0.51\textwidth}
    \includegraphics[width=\linewidth]{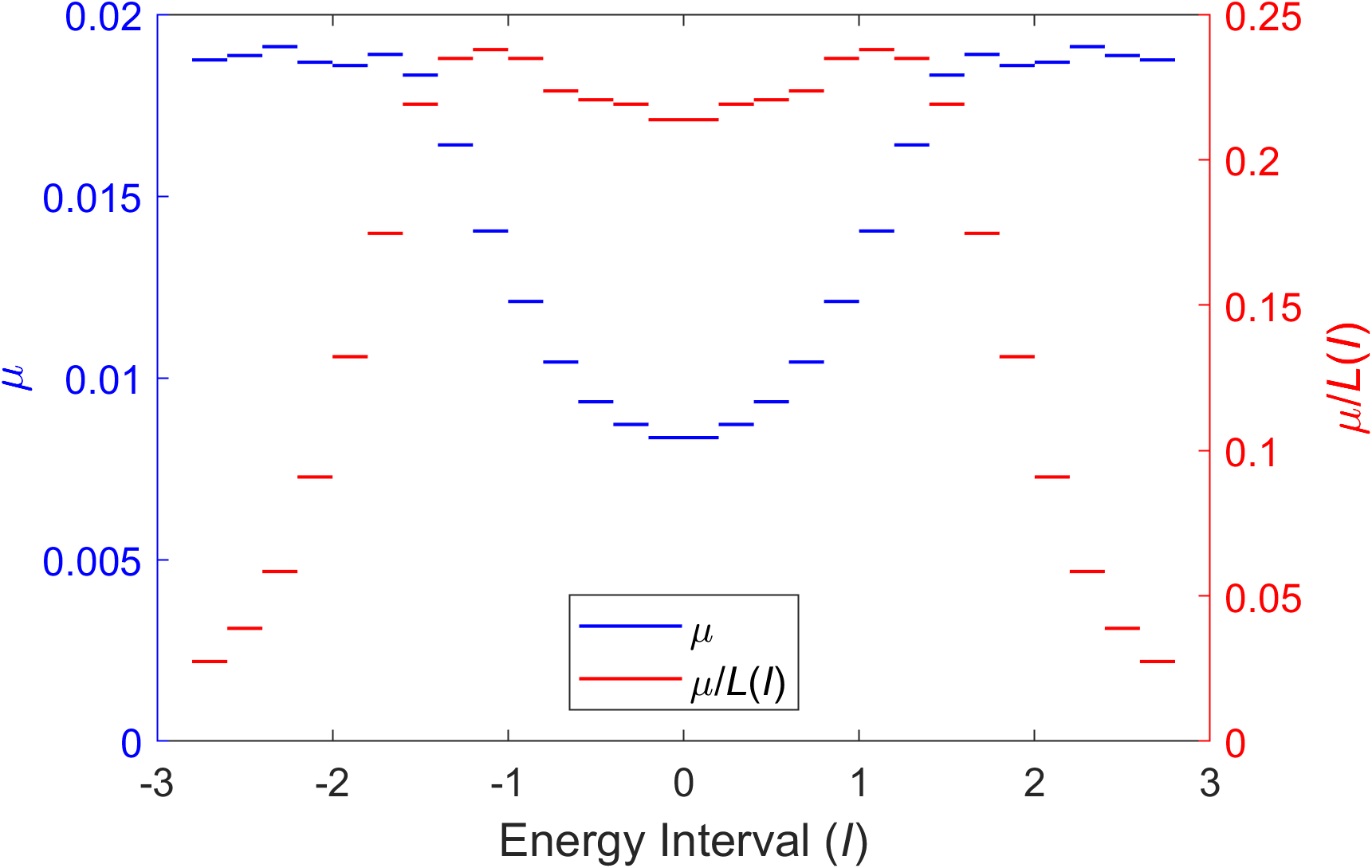} \caption{\label{Fig:exponents}
      On each energy interval with observed eigenvalues (see
      Fig.~\ref{Fig:Correlators}), the exponent $\mu$ at which
      $\nicefrac{C_2(\mu,I,\Lambda)}{C_2(0,I,\Lambda)}=2$ is shown in
      blue and the ratio $\nicefrac{\mu}{L(I)}$ is plotted in red,
      where $L(I)$ is the minimal Lyapunov exponent on $I$.}
  \end{subfigure}
  \hfill
  \begin{subfigure}[t]{0.45\textwidth}
    \includegraphics[width=\linewidth]{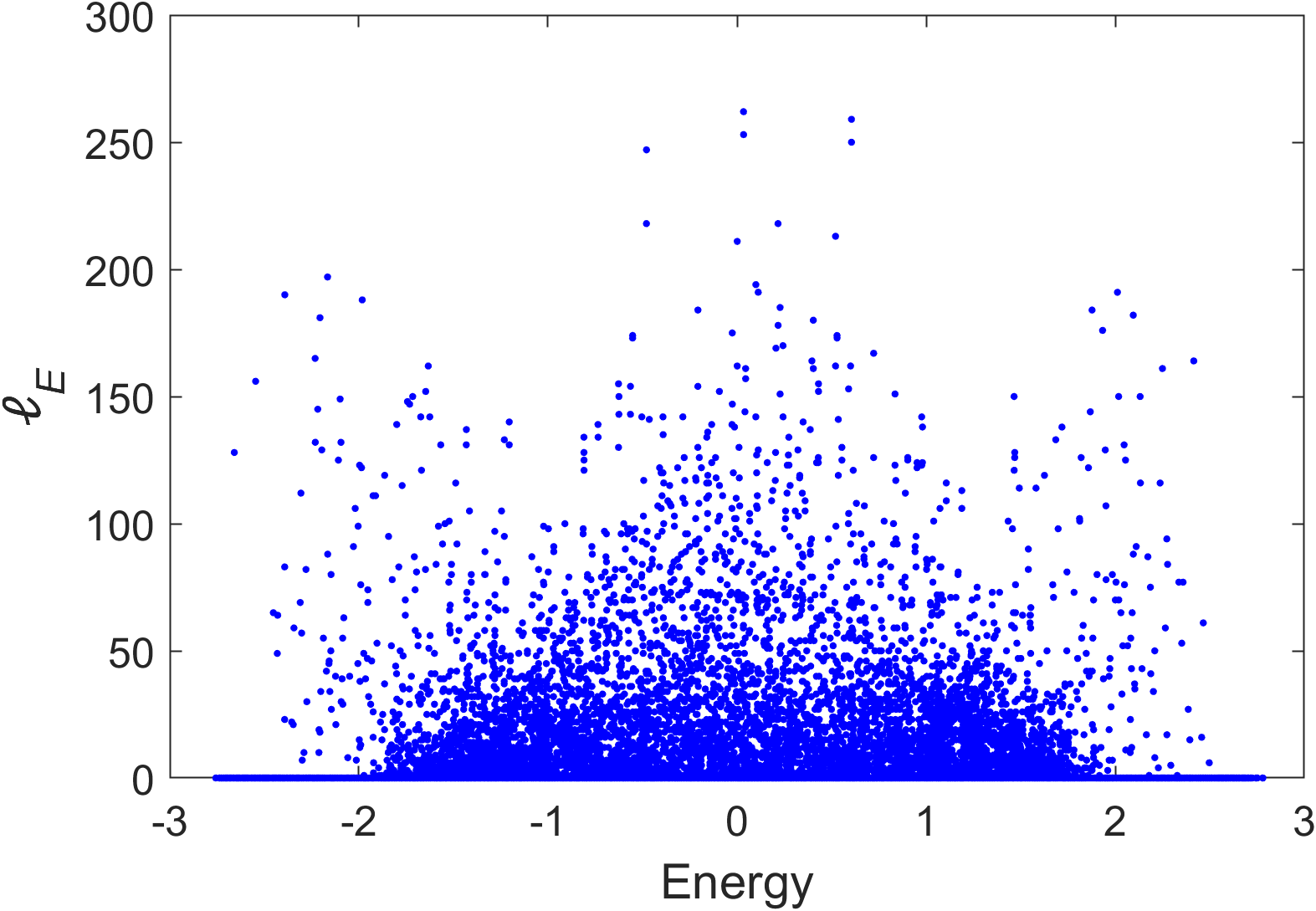} \caption{\label{Fig:onset}
      Onset length $\ell_E$ versus energy $E$ for the eigenfunctions
      from $240$ samples ($480,000$ eigenfunctions in total). Only
      $6,313$ (1.3\% of the total) eigenfunctions have $\ell_E>0$. }
  \end{subfigure}
  \caption{Exponents and onset length for eigenfunctions of the 1D
    Anderson model on an interval of 2000 with disorder $\lambda=1$..}
\end{figure}
In Fig. \ref{Fig:Correlators}, numerical estimates of the normalized
correlators $\nicefrac{C_2(\mu,I,\Lambda)}{C_2(0,I,\Lambda)}$ for
$(H_\omega)_\Lambda$ with $\lambda=1$ and $\Lambda=[1,2000]$ are shown
for 14 energy intervals and various values of $\mu$. These
computations were obtained by averaging results from 240 samples of
direct diagonalization of $(H_\omega)_\Lambda$.\footnote{To obtain
  these results, we used the Lyapunov exponent $L(3)$ at the edge of
  the spectrum to estimate the numerical precision with which to
  compute the eigenfunctions. We then used the GEM library \cite{GEM}
  to compute spectral data accurate to
  $\lceil \frac{L(3)}{\log 10}*2000 \rceil +5 = 810 $ decimal
  points. The logs of the eigenfunction densities,
  $\log |\varphi_j(x)|^2$, trimmed to double precision, were then used
  to compute the correlators and the onset lengths.}  Note that the
correlator blows up to \emph{extremely} large values
($2^{256} \approx 10^{77}$) well before $\mu$ approaches the Lyapunov
exponent \textemdash \ observe the $\log \log$ scale on the ordinate
of the plot!

\begin{figure}
  \centering
  \begin{subfigure}[t]{0.47\textwidth}
    \includegraphics[width=\linewidth]{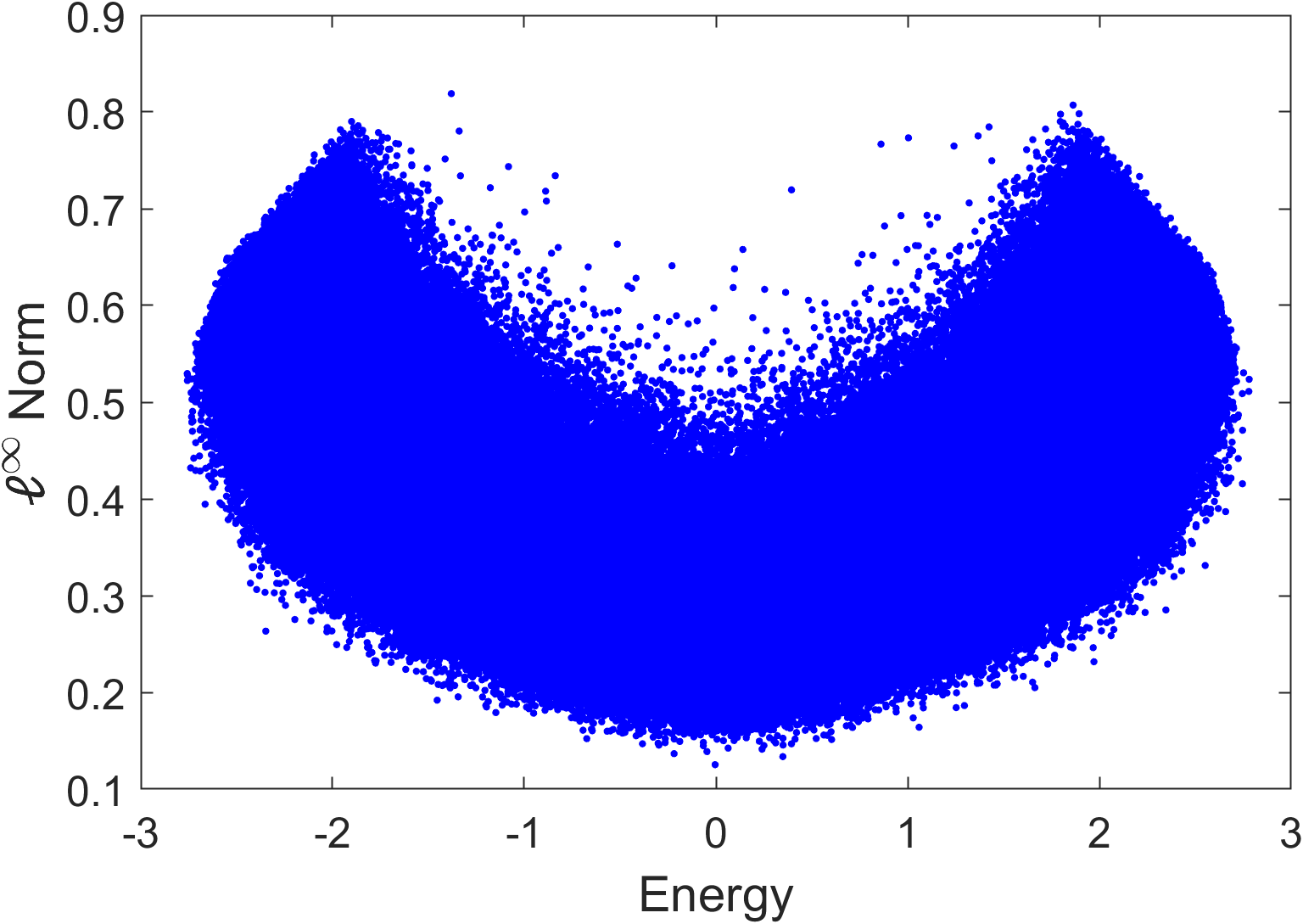}
    \caption{$\ell^\infty$ norm versus energy for the eigenfunctions
      from $240$ samples ($480,000$ eigenfunctions in total).  }
  \end{subfigure}
  \hfill
  \begin{subfigure}[t]{0.48\textwidth}
    \includegraphics[width=\linewidth]{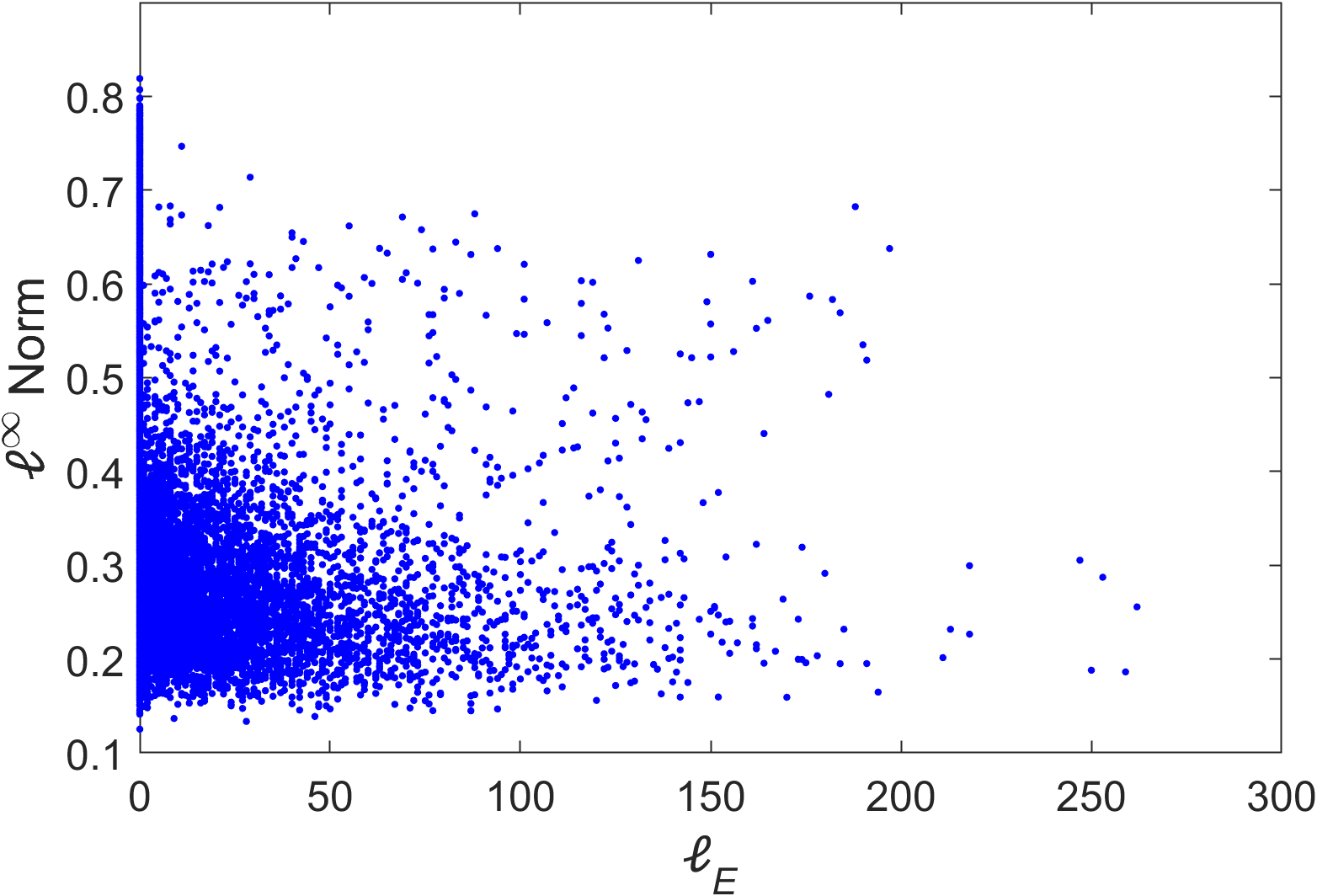}
    \caption{$\ell^\infty$ norm versus $\ell_E$ for the eigenfunctions
      from $240$ samples ($480,000$ eigenfunctions in total). }
  \end{subfigure}
  \caption{\label{Fig:ellinfinity}$\ell^\infty$ norms versus energy
    and onset length for eigenfunctions of the 1D Anderson model on an
    interval of 2000 with disorder $\lambda=1$.}
\end{figure}
To compute onset lengths for each energy interval, we chose an
exponent $\mu=\mu_I$ such that the correlator
$\nicefrac{C_2(\mu,I,\Lambda)}{C_2(0,I,\Lambda)} \approx 2$ (for
reference the horizontal cutoff at $2$ is shown in Fig.\
\ref{Fig:Correlators}).  These exponents are plotted for each of the
fourteen energy intervals from $0$ to $2.8$ in Fig.\
\ref{Fig:exponents}, along with the ratio $\mu/L(I)$ for each
interval.  In Fig.\ \ref{Fig:onset}, the onset length $\ell_E$ for
each of the $480,000$ eigenfunctions is plotted against the
corresponding eigenvalue.  Only $6,313$ eigenfunctions (1.3\% of the
total) were found to have positive onset length, and the maximum onset
length observed was 262.  For comparison, the $\ell^\infty$ norms of
the eigenfunctions are plotted in Fig.\ \ref{Fig:ellinfinity}, versus
energy and versus onset length.  Within the range of attained
$\ell^\infty$ norms, there is little correlation with the onset
length.

It is natural to wonder whether the observation of onset lengths as
large as 262 is consistent with the estimation that these numbers
should be ``of order $\log \#\Lambda \approx 7.6$.''  However, the
correlator bound $C_2(\mu,I,\Lambda)\le 2*C_2(0,I,\lambda)$ provides
an \emph{a priori} bound on localization lengths that is consistent
with this observation, as follows.  From the Markov inequality, we
have with probability at least $1-\epsilon$ that
$$\frac{1}{\#\Lambda}\sum_{x,y\in \Lambda} e^{2\mu |x-y|}  \sum_{E\in I\cap\Ec((H_\omega)_\Lambda)} |\varphi_E(x)|^2 |\varphi_E(y)|^2 \ \le \ \frac{2*C_2(0,I,\lambda)}{\epsilon}  .$$
Fixing a particular energy $E$ and taking $y=x_E$, we see that each
eigenfunction satisfies
$$\sum_{x} e^{2\mu|x-x_E|}  |\varphi_E(x)|^2 \ \le \ \frac{2*C_2(0,I,\Lambda) \# \Lambda}{\epsilon \|\varphi_E\|_\infty^2} \ . $$
Following the proof of Prop.\ \ref{proposition}, we find that
$$\ell_E \ \le \ \frac{1}{2 \mu} (\log 2 C_2(0,I,\Lambda) + \log \# \Lambda - \log 3 \epsilon -2 \log \|\varphi_E\|_\infty ) + 1.$$
In the current context, we take $\epsilon = 1/240$, as this is the
smallest probability we can resolve with 240 samples.  The key point
is that we expect onset lengths to be no larger than
$\frac{1}{2\mu}(\log \#\Lambda - \log 3\epsilon)$, where we have
neglected the relatively smaller terms coming from the
$\|\varphi_E\|_\infty$ norm and $C_2(0,I,\Lambda)$. For
$\#\Lambda=2000$, $\epsilon=1/240$, and $\mu \approx 0.01$, we obtain a
rough bound of order $600$.\footnote{We can include the contribution
  from $C_2$ and $\|\varphi_E\|_\infty$ as follows. Recall that
  $C_2(0,I,\Lambda)$ is equal to the integrated density of states on
  $\Lambda$ (see eq.\ \eqref{eq:C20=IDS}). For the energy intervals
  shown in Fig. \ref{Fig:Correlators} the maximum IDS of $0.1092$ is,
  which is attained for $I=\{2\le |E| \le 2.2\}$. The minimum
  $\ell^\infty$ norm over all intervals is $0.1253$ and the minimum
  value of $\mu$ is $0.0084$.  Taken together we have the the \emph{a
    priori} bound
  $\ell_E \ \le \ \frac{1}{2*0.0084} \left (\log 2*0.1092 + \log 2000
    + \log 80 - 2 \log 0.1253 \right ) +1 \ \approx \ 871 \ .$ This
  can be improved somewhat by computing separate bounds for each
  interval based on the exponent for that interval and $\ell^\infty$
  norms for eigenfunctions with energies in that interval, resulting
  in an upper bound of $844$.}  So we should not be surprised to see
onset lengths of the size seen here.

\begin{figure}
  \centering
  \begin{subfigure}[t]{0.48\textwidth}
    \includegraphics[width=\linewidth]{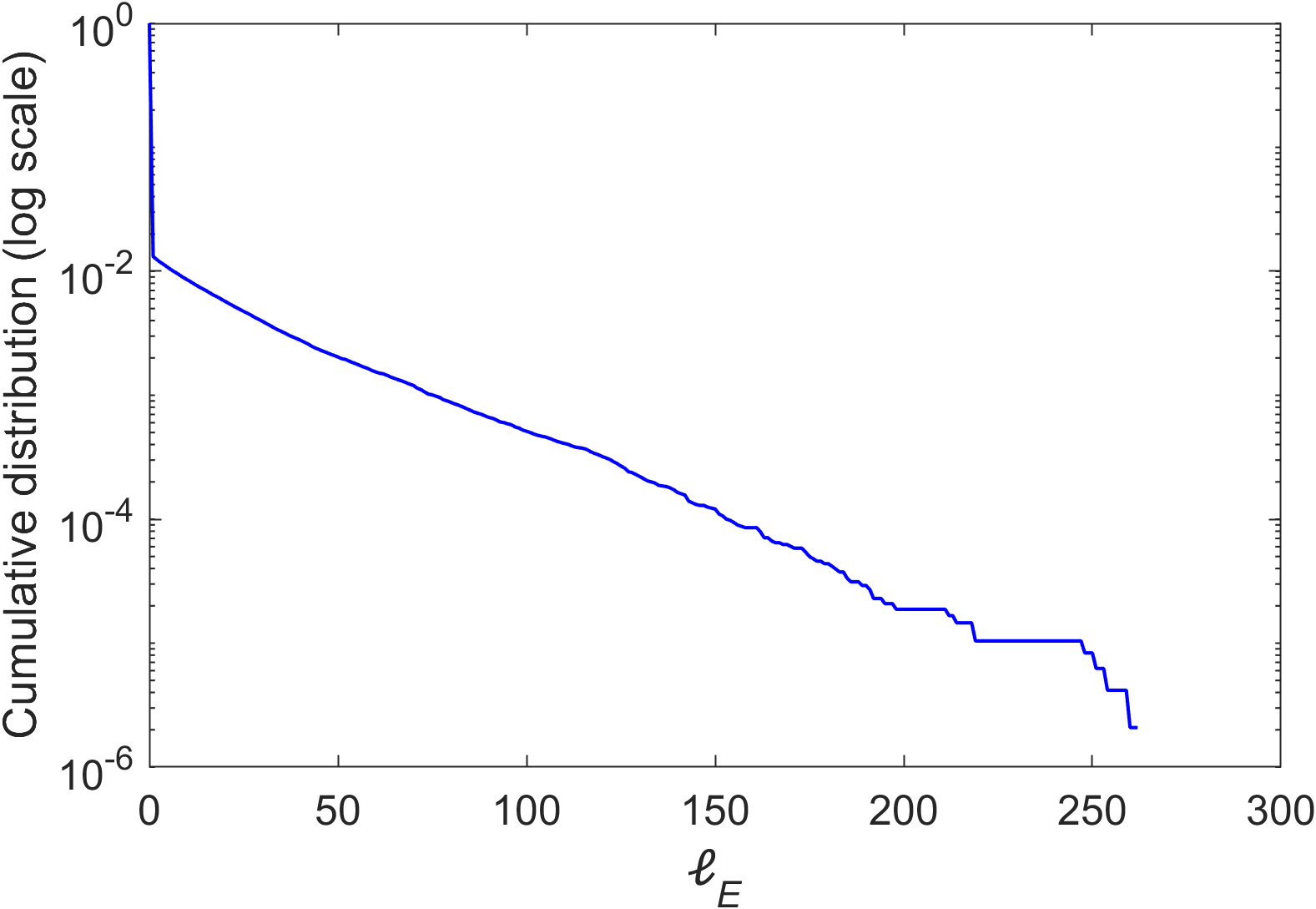} \caption{Cumulative
      distribution of onset lengths for $240$ samples with $\lambda=1$
      on an interval of length $2000$. }
  \end{subfigure}
  \hfill
  \begin{subfigure}[t]{0.48\textwidth}
    \includegraphics[width=\linewidth]{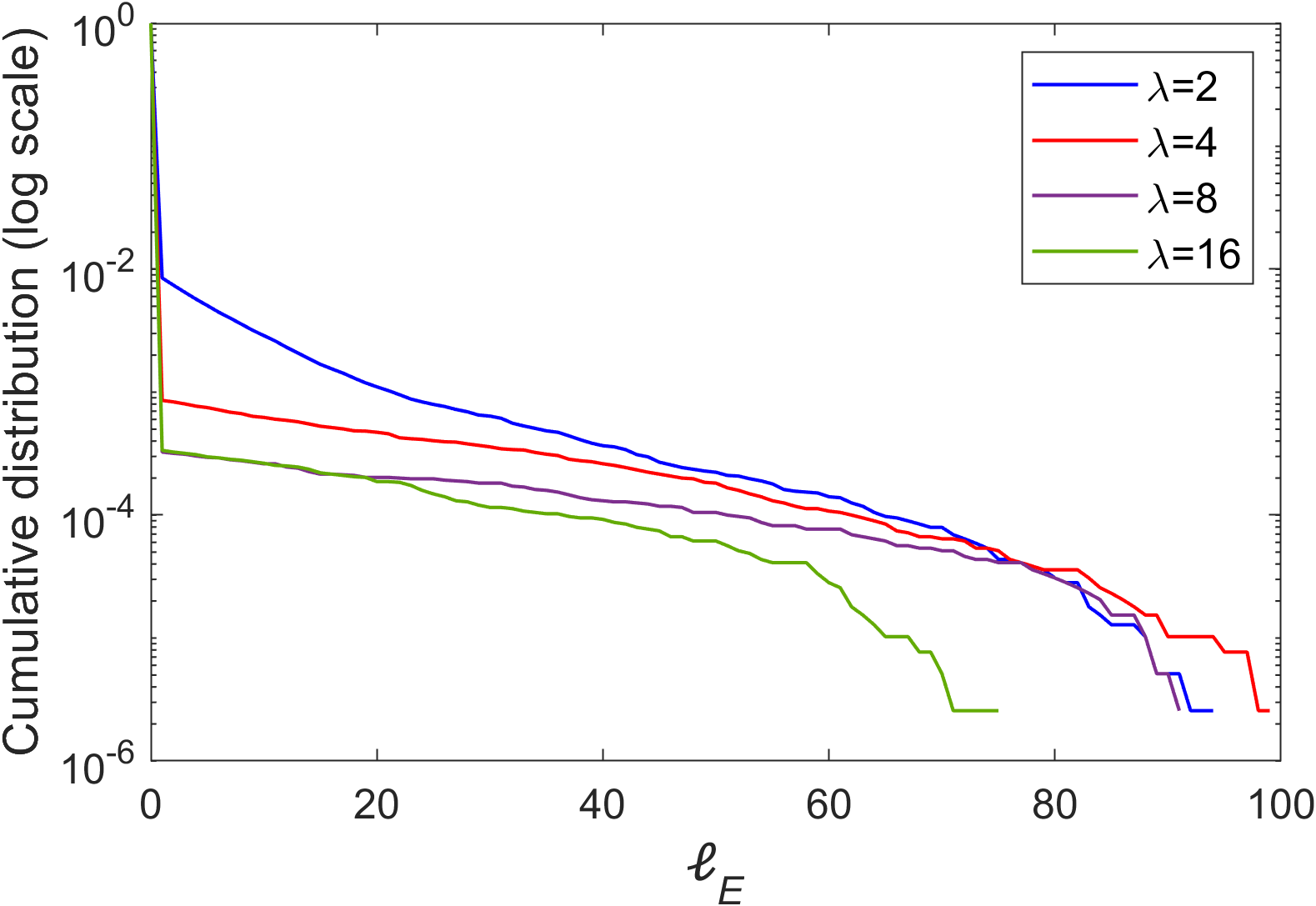} \caption{Cumulative
      distribution of onset lengths for $390$ samples with
      $\lambda=2$, $4$, $8$ and $16$ on an interval of length $1000$.
    }
  \end{subfigure}
  \caption{\label{Fig:cumulative}Cumulative distribution of onset
    lengths for eigenfunctions of the 1D Anderson model.}
\end{figure}
The cumulative distributions of onset lengths for various disorder
strengths are shown in Fig. \ref{Fig:cumulative}.  On the left, we
have plotted the results for the $\lambda=1$ and $\Lambda=[1, 2000]$.
On the right, one finds results for $\lambda=2$, $4$, $8$ and $16$ on
the interval $[1,1000]$, computed by the same methods indicated above.
Notably, for each disorder strength after a sharp drop off from
$\ell_E=0$ to $\ell_E=1$, the cumulative distribution exhibits
exponential decay over a range of lengths before dropping off at the
maximum attained onset length within the geometry and number of
simulations.

\begin{figure}
  \centering
  \includegraphics[width=\linewidth]{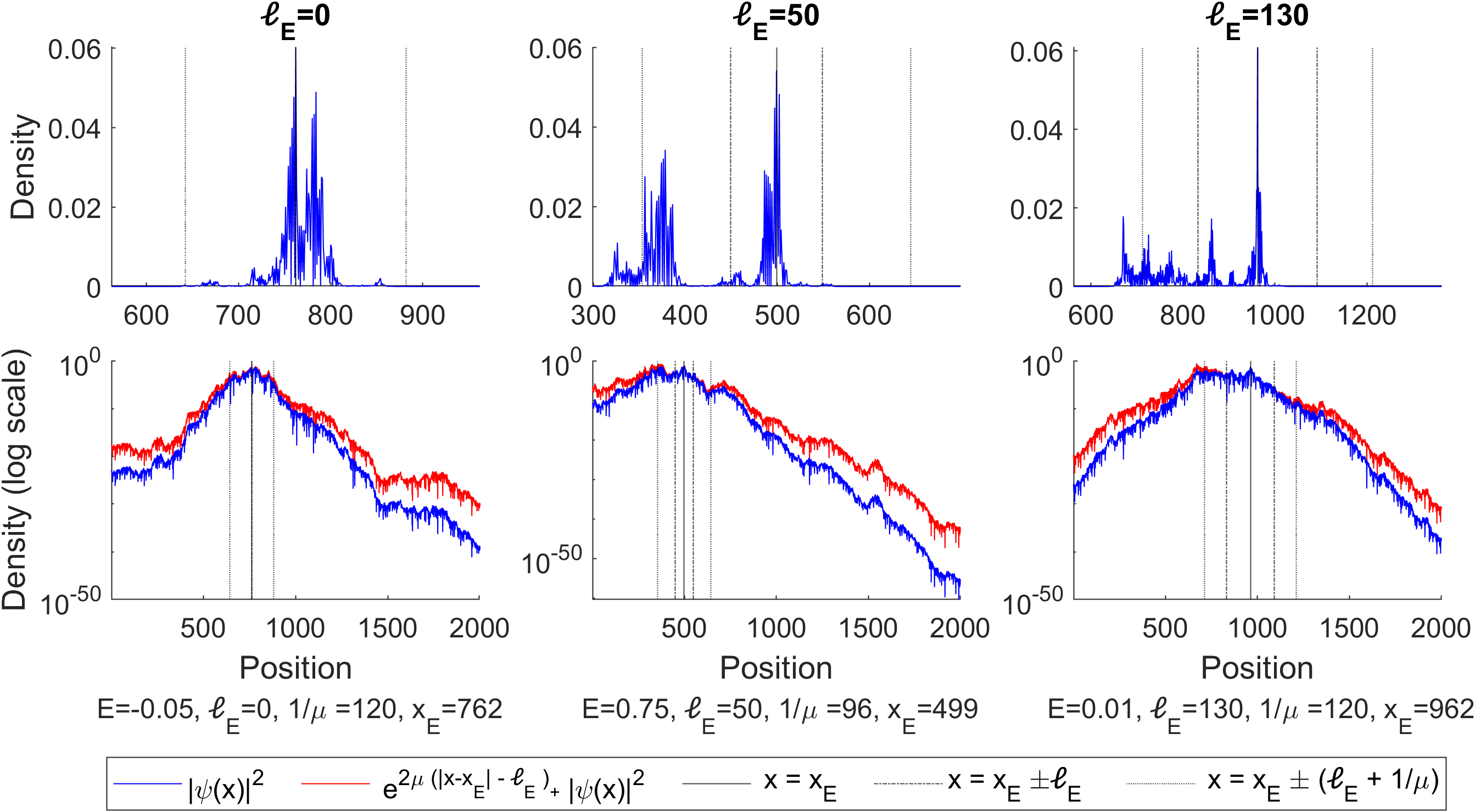}
  \caption{\label{Fig:various}Three digenfunctions for the 1D Anderson
    model on $\Lambda=[1,2000]$ with $\lambda=1$.}
\end{figure}
As the proof of Thm.\ \ref{thr:1} will show, an eigenfunction can have
a large onset length due to a large deviation of the random
environment in a neighborhood of the localization center.  As such,
although all eigenfunctions share the same exponential decay in their
tails, the behavior of an eigenfunction over the localization volume
$x_E\pm \ell_E$ is, by definition, \emph{atypical}.  To paraphrase
Tolstoy,\footnote{The opening of Anna Karenina: ``All happy families
  are alike; each unhappy family is unhappy in its own way.''}
\emph{all eigenfunctions with small onset length are alike, each
  eigenfunction with large onset length is extended in its own way.}
To illustrate the variety of behaviors possible within the
localization volume, we have have plotted the density $|\psi_E(x)|^2$
for three different eigenfunctions in Fig.\ \ref{Fig:various}.  In the
upper plots, the density of each eigenfunction is shown in a
neighborhood of the localization center.  In the lower plots, the
logarithm of the density and of the exponentially weighted density
$e^{2\mu(|x-x_E|-\ell_E)_+}|\psi_E(x)|^2$ are shown for the entire
chain $[1,2000]$.  The first eigenfunction, with onset length
$\ell_E=0$, has the majority of its mass within one localization
length ($\nicefrac{1}{\mu}$) of the localization center.  The second
eigenfunction, with onset length $\ell_E=50$, shows two distinct peaks
at roughly distance $\ell_E+\nicefrac{1}{\mu}$ from each other.  This
sort of resonant superposition of two or more localization centers is
one mechanism for the development of a substantial onset length.
Although rare, such eigenfunctions will appear with positive frequency
in large systems.  Finally, the third eigenfunction, with onset length
$\ell_E=130$, is extended over an interval of size roughly
$\ell_E+1/\mu$.  On mechanism for the occurrence of such
eigenfunctions is for the potential over the interval
$[x_E-\ell_E,x_E+\ell_E]$ to closely mimic a potential having extended
states (e.g., a periodic potential) or a long localization length at
energy $E$.  This is the behavior one finds for the eigenfunctions
from the Lifshitz tail regime (see Thm.\ \ref{thr:3}), although this
particular eigenfunction comes from the center of the band ($E=0.01$).

\section{The proofs of Theorem \ref{thr:1} and Theorem \ref{thr:8}}
\label{sec:proof-theorem}
We now turn to the proof of Theorem \ref{thr:1} for the discrete
Anderson model.  Let $\alpha \in \frac{1}{2}\N$, $\alpha \ge 1$; this parameter is
arbitrary but will be fixed throughout our analysis.  Given a lattice
cube $Q=\Lambda_{\ell}(x_0)$ of center $x_0\in \Z^d$ and side length
$\ell\in \N$, we let $\widetilde{Q} := \Lambda_{(2\alpha+1)\ell}(x_0)$ denote
the expanded cube with the same center but side length $(2\alpha +1) \ell$
(see Figure \ref{Fig:setup}).  Our first result shows how to
approximate an eigenfunction $\varphi_{E}$ on a region $\Omega$ with
localization center $x_E\in Q$ by an eigenfunction
$\psi_{E'}\in \ell^2(\widetilde{Q})$ of the Hamiltonian
$(H_\omega)_{\widetilde{Q}}$ on the expanded cube.

\begin{Le} \label{le:1} Let $\Omega \subset \Z^d$ be a region and let $S=\Omega \cap Q$,
  where $Q= \Lambda_{\ell}(x_0)$ is a cube of side length $\ell\ge 2$
  such that $Q\cap \Omega \neq \emptyset$. Let $\widetilde{S}=\widetilde{Q} \cap \Omega $,
  let $0 < \delta < 1 $ and fix a disorder configuration $\omega$ such that
  \begin{equation}\label{eq:minamilemma}
    \max_{E\in\sigma \left ((H_\omega)_{\widetilde{S}} \right )}\tr\car_{E+[-\delta,\delta]}
    [(H_\omega)_{\widetilde{S}}] \ \leq \ 1 \ .
  \end{equation} 
  If there is an eigenvalue $E\in \Ec((H_\omega)_{\Omega})$ of
  $(H_\omega)_{\Omega}$ with normalized eigenvector $\varphi_{E}$ such
  that $x_{E} \in S$ and
  \begin{equation}\label{eq:Mbound}
    M \ := \  \left ( \sum_{x} e^{2\mu|x-x_E|} |\varphi_{E}(x)|^2  \right )^{\frac{1}{2}} \ \le  \frac{\delta e^{\alpha \mu \ell }}{2 \sqrt{d}} \ ,
  \end{equation}
  then there exists a unique eigenvalue
  $E' \in \Ec((H_\omega)_{\widetilde{S}})$ of
  $(H_\omega)_{\widetilde{S}}$ with normalized eigenvector $\psi_{E'}$
  such that
  \begin{equation}\label{eq:lemmaE-E'}
    |E-E'| \ \le \   2 \sqrt{d} M e^{-\alpha \mu \ell } \ ,
  \end{equation}
  and
  \begin{equation}\label{eq:lemmaphi-psi}
    \left\|\car_{\widetilde{S}} \varphi_{E} -\psi_{E'} \right\|_{\ell^2(\widetilde{S})} \ \le \    3\sqrt{d} \frac{M}{\delta}  e^{-\alpha \mu  \ell} \ .
  \end{equation}
\end{Le}
\begin{figure}
  \centering \includegraphics[width=0.5\textwidth]{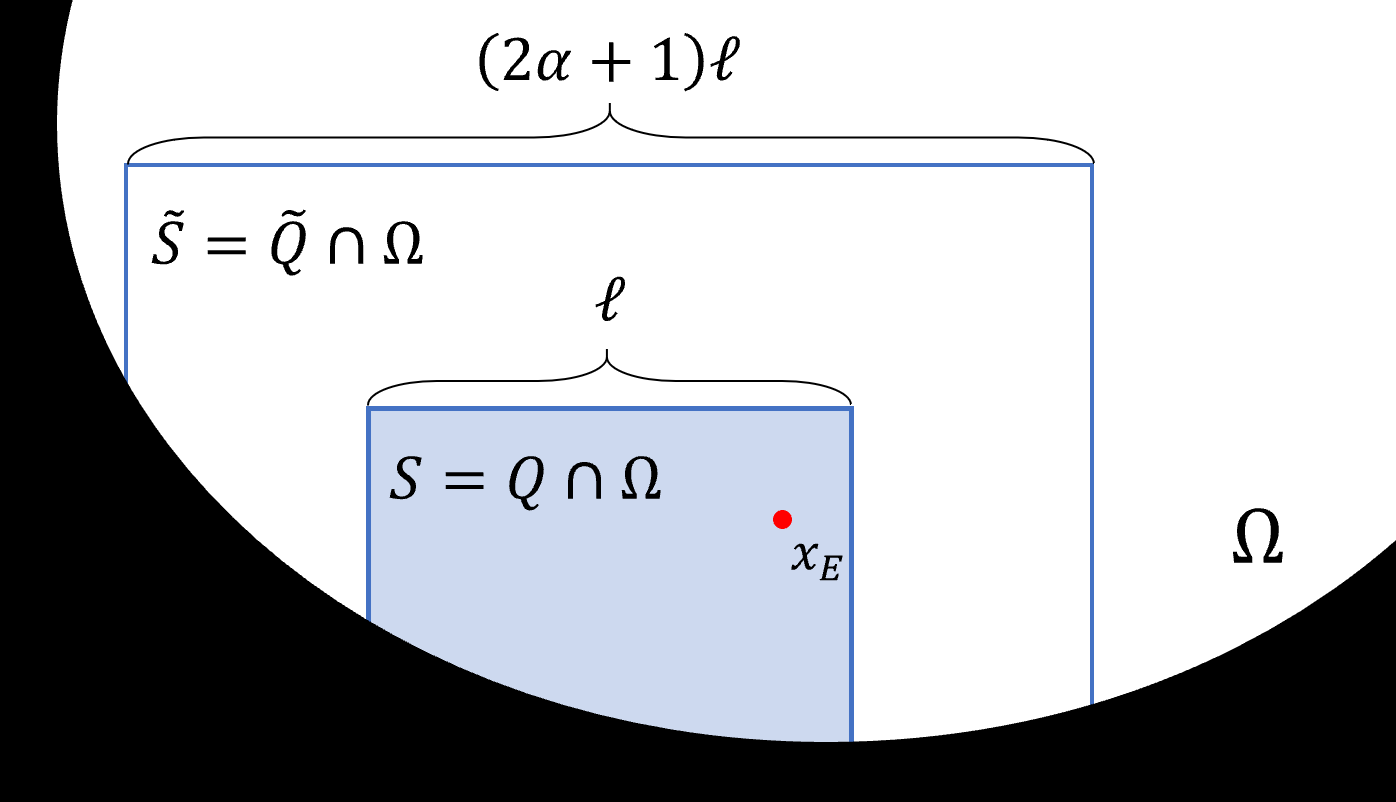}
  \caption{\label{Fig:setup}The setup of Lemma \ref{le:1}, showing a
    cube $Q$ and its expansion $\widetilde{Q}$ intersected with a
    region $\Omega$. }
\end{figure}
\begin{proof} Since $x_{E}\in S$, we have
  \begin{equation}\label{eq:phiEn-1Q'Q}
    \|\varphi_{E}\|_{\ell^2(\Omega \setminus \widetilde{S})} \ 
    = \ \left ( \sum_{x\in \Omega \setminus \widetilde{S}} |\phi_E(x)|^2 \right )^{\frac{1}{2}} \ \le \ \ e^{-\alpha \mu \ell} \left ( \sum_{x\in \Omega } e^{2\mu|x-x_E|}|\phi_E(x)|^2 \right )^{\frac{1}{2}} \
    = \  M e^{-\alpha \mu \ell} \ .
  \end{equation}
  Thus by \eqref{eq:Mbound},
  \begin{equation}\label{eq:phiEn-1Q}
    \| \car_{\widetilde{S}} \varphi_{E}\|_{\ell^2(\widetilde{S})} \ \ge \ 1 -  M e^{-\alpha \mu \ell}   \ \ge \ 1 - \frac{\delta }{2 \sqrt{d}} \ > \ \frac{1}{2} \ .
  \end{equation}
  Since $(H_\omega)_{\widetilde{S}}$ is the restriction of
  $(H_\omega)_{\Omega}$ to $\widetilde{S}$ and $\varphi_{E}$ is an
  eigenvector of $(H_\omega)_{\Omega}$, we see that
  $$ \left [ (H_\omega)_{\widetilde{S}}- E \right ] \car_{\widetilde{S}} \varphi_{E}(x) = \begin{cases}
    0 & \text{if } x\in \widetilde{S} \text{ and } \dist(x,\Omega\setminus \widetilde{S})\ge 2 \ , \text{ and} \\
    {\D \sum_{\substack{y\in \Omega\setminus \widetilde{S} \\
          |y-x|=1}}} \varphi_{E}(y) & \text{if } x\in \widetilde{S}
    \text{ and } \dist(x,\Omega\setminus \widetilde{S})=1 \ .
  \end{cases}$$ Because each $x\in \widetilde{S}$ has at most $d$
  neighbors in $\Omega \setminus \widetilde{S}$,
  \begin{equation}
    \label{eq:37}
    \left\|\left[(H_\omega)_{\widetilde{S}}-E\right]\car_{\widetilde{S}}
      \varphi_{E}\right\|
    \ \le \  \left (d \sum_{y\in \Omega\setminus \widetilde{S}} |\varphi_E(y)|^2\right)^{\frac{1}{2}}  \ \le \ \sqrt{d}  M e^{-\alpha \mu  \ell} .
  \end{equation}
  Since $(H_\omega)_{\widetilde{S}}$ is self adjoint, \eqref{eq:37}
  and \eqref{eq:phiEn-1Q} together imply that one of the eigenvalues
  of $(H_\omega)_{\widetilde{S}}$, call it $E'$, satisfies
  \eqref{eq:lemmaE-E'}.
	
  Let $\psi_{E'}$ be the normalized eigenvector associated to $E'$; we
  fix its phase by requiring
  $ \langle\psi_{E'} \, , \, \car_{\widetilde{S}} \varphi_{E} \rangle
  > 0$.  To estimate
  $\left\|\car_{\widetilde{S}}\varphi_E -\psi_{E'}
  \right\|_{\ell^2(\widetilde{S})}$, note that
  \begin{multline}\label{eq:3finally}
    \left\|\car_{\widetilde{S}}\varphi_E -\psi_{E'}
    \right\|_{\ell^2(\widetilde{S})} \ \le \ \| \car_{\widetilde{S}}
    \varphi_{E} - \langle\psi_{E'} \, , \, \car_{\widetilde{S}}
    \varphi_{E} \rangle \psi_{E'} \|_{\ell^2(\widetilde{S})} + 1-
    \langle\psi_{E'} \, , \, \car_{\widetilde{S}} \varphi_{E} \rangle
    \\ \le \ 2 \| \car_{\widetilde{S}} \varphi_{E} - \langle \psi_{E'}
    \ , \ \car_{\widetilde{S}} \varphi_{E} \rangle \psi_{E'}
    \|_{\ell^2(\widetilde{S})} + \|\varphi_{E}\|_{\ell^2(\Omega
      \setminus \widetilde{S})},
  \end{multline}
  since
  $\langle\psi_{E'} \, , \, \car_{\widetilde{S}} \varphi_{E} \rangle
  \ge \| \car_{\widetilde{S}} \varphi_{E}\|_{\ell^2(\widetilde{S})} -
  \| \car_{\widetilde{S}} \varphi_{E} - \langle\psi_{E'} \, , \,
  \car_{\widetilde{S}} \varphi_{E} \rangle \psi_{E'}
  \|_{\ell^2(\widetilde{S})} $.  By~\eqref{eq:minamilemma}, $E'$ is
  non-degenerate and at least distance $2\delta$ from every other
  eigenvalue of $(H_\omega)_{\widetilde{S}}$.  Thus it follows from
  \eqref{eq:37} and \eqref{eq:lemmaE-E'} that
  \begin{equation}\label{eq:halfwayto3}
    \| \car_{\widetilde{S}} \varphi_{E} - \langle\psi_{E'} \, , \, \car_{\widetilde{S}} \varphi_{E}   \rangle \psi_{E'} \|  \ \le \  \frac{ \sqrt{d}  M e^{-\alpha \mu  \ell}}{2 \delta - 2 \sqrt{d}  M e^{-\alpha \mu  \ell } } 
    \ \le \  \  \sqrt{d}\frac{M}{\delta} e^{-\alpha \mu  \ell} \ ,
  \end{equation}
  where we have used \eqref{eq:Mbound} in the last step.  Equation
  \eqref{eq:lemmaphi-psi} follows from \eqref{eq:3finally},
  \eqref{eq:halfwayto3}, and \eqref{eq:phiEn-1Q'Q}.
\end{proof}
For a given sufficiently large length scale $L$ and a region $\Omega$
containing a cube $\Lambda_L(x_0)$ of size $L$, we will consider
partitions of $\Omega$ into smaller cubes along a decreasing sequence
of scales, depending on $L$:
\begin{equation}\label{eq:Ln}
  L_1 = \lfloor  \beta \log L \rfloor \quad \text{and} \quad  L_n = \lfloor \tfrac{1}{4} L_{n-1}   \rfloor \quad \text{for }2 \le n \le n_{\mathrm{fin}}  \ .
\end{equation}
Here $\beta > 0$ is a constant to be chosen below and
$n_{\mathrm{fin}}=n_{\mathrm{fin}}(L)$ is the largest value of $n$
such that $L_n\ge L_{\mathrm{fin}}$, where $L_{\mathrm{fin}}>1$ is an
(integer) length scale which we take sufficiently large, but fixed
independent of $\Omega$ and $L$.  In particular, we require that
$L_{\mathrm{fin}}\ge e$ and
$L_{\mathrm{fin}}>\beta \log L_{\mathrm{fin}}$, so that $L_1=\beta \log L < L$
for $L>L_{\mathrm{fin}}$.  Without loss of generality, we suppose that
$L\ge \exp(\nicefrac{L_{\mathrm{fin}}}{\beta})$ so that
$L>L_1\ge L_{\mathrm{fin}}$ and $n_{\mathrm{fin}}\ge 1$. Note that
$ L_{\mathrm{fin}} \le L_{n_{\mathrm{fin}}}< 4 L_{\mathrm{fin}}$ and
\begin{equation}\label{eq:Lnbounds}
  (\tfrac{1}{4})^{n-1} \beta \log L- \frac{4}{3} \ < \ L_n \ \le \ (\tfrac{1}{4})^{n-1} \beta \log L.
\end{equation}
As a result, we have the estimate
\begin{equation}
  \label{eq:5}
  \frac{1}{\log 4}\log \log L +\frac{
    \log \beta -\log( L_{\mathrm{fin}}+\frac{1}{3})}{\log 4} \ <  n_{\mathrm{fin}}(L) \ \le \  \frac{1 }{\log 4}\log \log L +\frac{\log \beta -\log L_{\mathrm{fin}}}{\log 4} +1 \ .
\end{equation}
In particular $n_{\mathrm{fin}}(L)= O(\log \log L)$ as
$L\rightarrow \infty$.  For future reference we note the following
\begin{Pro}\label{pro:scalesum}For every $\nu >0$ and
  $1\le n\le n_{\mathrm{fin}}$,
  \begin{equation}\label{eq:scalesum}
    e^{-\nu L_n} \ < \
    \sum_{j=1}^{n} e^{-\nu L_j} \ < \  \frac{e^{-\nu L_{n}}}{1-e^{-\nu L_{n}}} \ < \ \left ( 1 + \frac{1}{e^{\nu L_{\mathrm{fin}}}-1} \right ) e^{-\nu L_n} \ . 
  \end{equation} 
\end{Pro}
\begin{proof} Since $L_j \ge 4 L_{j+1}$, we see that
  $L_j \ge 4^{n-j} L_n$ for $1\le j \le n$. Thus
  $$\sum_{j=1}^{n} e^{-\nu L_j} \ \le \ \sum_{j=0}^{n-1} e^{-\nu
    4^{j} L_{n}} \ < \ \sum_{j=0}^{\infty} e^{-\nu 4^{j} L_n} \ < \
  e^{-\nu L_n} \sum_{j=0}^\infty e^{-3 j \nu L_n} \ < \ e^{-\nu L_n}
  \sum_{j=0}^\infty e^{- j \nu L_n} \ , $$ from which the upper bounds
  follow. The lower bound is clear.
\end{proof}
We now fix a region $\Omega\subset \Z^d$ and a length scale $L$. For each
generation $1\leq n\leq n_{\mathrm{fin}}$, let
$$\mathcal{G}^n \ := \ \left \{ \Lambda_{L_n}(L_n\mb{k}) \cap \Omega \ :
  \ \mb{k}\in \Z^d \text{ and } \Lambda_{L_n}(L_n\mb{k}) \cap \Omega\neq
  \emptyset \right \} \ , $$
\begin{figure}
  \centering{}
  \includegraphics[width=0.5\textwidth]{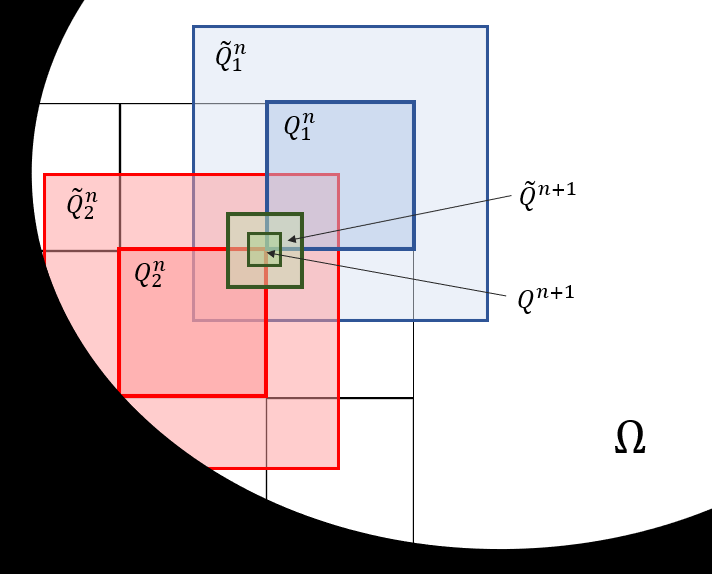} \caption{\label{Fig:generations}
    Cubes of generation $n$ and $n+1$.  Note that the neighboring
    cubes $Q_1^n$ and $Q_2^n$ do not overlap, but the extended cubes
    $\widetilde{Q}_1^n$ and $\widetilde{Q}_2^n$ do.}
\end{figure}
which is the set of cubes centered on $L_n \Z^d$ of side length $L_n$
which overlap $\Omega$ (see Figure \ref{Fig:generations}). (We shall
refer to the elements of $\mathcal{G}^n$ as ``cubes,'' although those
that intersect the boundary of $\Omega$ consist only of a portion of a
cube.)
Note that a cube $Q\in \Gc^n$ and its expansion $\widetilde{Q}$ have
volumes
\begin{equation}
  \label{eq:sizeofQ}
  \hskip-7cm
  \# Q\ \le \ L_n^d \quad \text{ and } \quad \# \widetilde{Q} \ \le \ (2\alpha +1)^d L_n^d \ ,
\end{equation}
with equality unless the cube $Q$ or its expansion $\widetilde{Q}$
intersect the boundary of $\Omega$. With these definitions, we see
that
\begin{enumerate}
\item If $Q,Q'\in \Gc^n$ and $Q\neq Q'$, then $Q\cap Q'=\emptyset$.
\item $\D \Omega \ = \ \bigcup_{Q\in \Gc^n} Q $.
\end{enumerate}
Given a region $S\subset \Omega$, let
\begin{equation*}
  \hskip-7cm \Gc^n(S) \ := \ \left \{ Q\in \Gc^n \ : \ Q\cap S \neq \emptyset
  \right \}.
\end{equation*}
We note that if $S=\Lambda_L(x) \subset \Omega$ is a cube then we have the
bounds:
\begin{equation}\label{eq:Njnbound}
  \hskip-8cm  \left  (\frac{L}{L_n}\right)^d \ \le \  \# \mathcal{G}^n(\Lambda_L(x)) \ \le \  \left( \frac{L}{L_n}+2 \right )^d
\end{equation}
%
Given a cube $Q \in \Gc^n$ of generation $n$ and a realization $\omega$
of the random potential, let
\begin{equation*}
  M_\omega(Q) \ = \ \max_{E\in  \Sigma(\widetilde{Q})} \max_{y
    \in \mathcal{C}(\varphi_E)}  \left ( \sum_{x\in \widetilde{Q}}
    |\varphi_E(x)|^2 e^{2\mu |x-y|} \right )^{\frac{1}{2}} \ ,
\end{equation*}
where
$\Sigma(\widetilde{Q})=\Ec((H_\omega)_{\widetilde{Q}})\cap
I_{\text{AL}}$ and $\varphi_E$ is the $\ell^2$-normalized eigenvector
of of $(H_\omega)_{\widetilde{Q}}$ corresponding to the eigenvalue
$E\in \Sigma(\widetilde{Q})$. For $n\ge 1$, we say that $Q\in \Gc^n$ is {\it
  $\epsilon$-good} (for a given realization $\omega$) if the following
conditions are satisfied:
\begin{enumerate}
\item $M_\omega(Q)\le e^{\epsilon \lfloor \frac{1}{4} L_{n}\rfloor}$;
\item Eq.\ \eqref{eq:minamilemma} holds for $\widetilde{Q}$ with
  $\delta= e^{-\frac{\epsilon}{2} L_{n}}$, i.e.,
  \begin{equation}
    \label{eq:36}
    \max_{E\in\sigma \left ((H_\omega)_{\widetilde{Q}} \right )\cap
      I_{\mathrm{AL}}}\tr\car_{E+[-e^{-\frac{\epsilon}{2} L_{n}},e^{-\frac{\epsilon}{2} L_{n}}]}
    [(H_\omega)_{\widetilde{Q}}] \ \leq \ 1. 
  \end{equation}
\end{enumerate}
Note that the exponent in (1) involves
$\lfloor \frac{1}{4}L_n \rfloor=L_{n+1}$, which is the next length
scale. By a bound similar to \eqref{eq:phiEn-1Q'Q}, if $Q\in \Gc^n$ is an
$\epsilon$-good cube and $\ell >0$, then
\begin{equation}\label{eq:awayfromcenter}
  \|\car_{\{x:|x-x_E|\ge \ell\}} \varphi_E\|_{\ell^2(\widetilde{Q})} \ \le \ e^{-\mu \ell + \epsilon \lfloor \frac{1}{4} L_n \rfloor}
\end{equation}
for any $E\in \Sigma(\widetilde{Q})$ and $x_E\in \mathcal{C}(\varphi_E)$.
The cube $\widetilde{Q}$ is called {\it $\epsilon$-bad} (for a given
realization $\omega$) if it is not $\epsilon$-good.

By iterating Lemma \ref{le:1}, we can obtain the following:
\begin{Le}\label{le:iteration} Let
  $0<\epsilon \le \frac{\alpha}{8}\mu$, let $p>0$, and let
  $\beta \ge \frac{4}{\alpha \mu} (p+d) $. Let $[a,b]\subset I_{\mathrm{AL}}$ be a
  compact interval with $r=\dist([a,b],I_{\mathrm{AL}}^c) >0$, let
  $L_{\mathrm{fin}}$ be such that
  $L_{\mathrm{fin}}>\beta \log L_{\mathrm{fin}}$,
  \begin{subequations}\label{eq:Lfinbounds}
    \begin{equation}
      L_{\mathrm{fin}} \ \ge \ \max \left ( 4, \ \frac{8}{5\alpha\mu}\log 2, \  \frac{8}{7\alpha\mu} \log \tfrac{4\sqrt{d}}{r} \right ) \,
    \end{equation} and
    \begin{equation}
      \sup_{L\ge  \frac{1}{32} L_{\mathrm{fin}} }  (4L+1)^{\frac{d}{2}}e^{- \mu L} \ \le \ \frac{1}{8\sqrt{d}},
    \end{equation}
  \end{subequations}
  and let $L \ge \exp(\nicefrac{L_{\mathrm{fin}}}{\beta})$.  If $\Omega \subset \Z^d$
  is a region and $\Lambda_L(x_0)\cap \Omega\neq \emptyset$ for some
  $x_0$, then, with probability at least
  $1-(e^{\frac{\alpha\mu}{8}} A_{\mathrm{AL}})^2L^{-p}$, to each eigenvalue
  $$E\ \in \ \Sigma(\Lambda_L(x_0))\ := \ \{ E \in
  \Ec((H_\omega)_\Omega)\cap [a,b] \ : \
  \mathcal{C}(\varphi_E)\cap\Lambda_L(x_0)\neq \emptyset \}$$ are
  associated finite sequences $(Q_E^j)_{j=1}^{m_E}$,
  $(\psi_E^j)_{j=0}^{m_E}$, and $(\lambda_E^j)_{j=0}^{m_E}$ where:
  \begin{enumerate}
  \item $\psi_E^0=\varphi_E$ and $\lambda_E^0= E$;
  \item if $m_E\ge 1$, then, for $j=1,\ldots,m_E$,
    \begin{enumerate}
    \item the cube $Q_E^{j}$ is $\epsilon$-good and
      $\mathcal{C}(\psi_E^{j-1})\cap Q_E^{j} \neq \emptyset$;
    \item if $j\ge 2$, then
      $\widetilde{Q}_E^{j} \subset \widetilde{Q}_E^{j-1}$;
    \item $\psi_E^{j}$ is an eigenfunction of
      $(H_\omega)_{\widetilde{Q}_E^{j}}$ with eigenvalue
      $\lambda_E^{j}$;
    \item we have
      \begin{align}|\lambda_E^j - \lambda_E^{j-1}| \ &\le \ 2 \sqrt{d} \, e^{-(\alpha \mu- \epsilon) L_j} , \label{eq:lambdaj}\\
        \intertext{and} \|\psi_{E}^j-\car_{\widetilde{Q}_E^j}
        \psi_E^{j-1} \|_{\ell^2(\widetilde{Q}_E^j)} \ &\le \ 3
                                                        \sqrt{d} \,
                                                        e^{-(\alpha
                                                        \mu-3\epsilon
                                                        ) L_j} \
                                                        ; \label{eq:psij}
      \end{align}
    \end{enumerate}
  \item either $m_E=n_{\mathrm{fin}}$ or every cube $Q\in \Gc^{m_E+1}$
    with $Q\cap \mathcal{C}(\psi_E^j)\neq \emptyset$ and
    $\widetilde{Q}\subset \widetilde{Q}_E^{m_E}$ is $\epsilon$-bad.
  \end{enumerate}
  Furthermore, taking $Q_E^0 =\widetilde{Q}_E^0=\Omega$ and $L_0=L$,
  we have
  \begin{enumerate}[label=(\alph*)]
  \item given integers $0\le j_E\le m_E$ for each
    $E\in \Sigma(\Lambda_L(x_0))$, the map
    $ E \mapsto (Q_{j_E},\psi_E^{j_E},\lambda_{j_E})$ is one-to-one,
  \item for any $y\in \mathcal{C}(\psi_E^{m_E})$,
    \begin{equation}\label{eq:centersclose}
      \mathcal{C}(\psi_E^j)\subset \{ x : |x-y|< \frac{\alpha}{16} L_{m_E}\} \subset \widetilde{Q}_{E}^{m_E} \ ,
    \end{equation}
    for each $j=0,\ldots,m_E$, and
  \item if $(1+\alpha)\nu < \alpha \mu -4\epsilon$, then
    \begin{equation}\label{eq:improvedbound}
      \left (\sum_{x\in \widetilde{Q}_E^{j}} e^{2\nu|x-y|}|\psi_E^{j}(x)|^2\right )^{\frac{1}{2}} \ \le \  \left (1 +  \frac{1+3\sqrt{d}}{e^{\epsilon L_{m_E}}-1}  \right ) e^{(\frac{\alpha}{8}+ \frac{\epsilon}{4})\nu L_{m_E}} 
    \end{equation}
    for any $y\in \mathcal{C}(\psi_E^j)$.
  \end{enumerate}
\end{Le}
\begin{Rem} 1) Since $L_{\mathrm{fin}} \ge \frac{8}{5\alpha \mu} \log 2$,
  Proposition \ref{pro:scalesum} implies that
  \begin{equation}\label{eq:scalesum2}
    \sum_{j=1}^{n} e^{-\nu L_j} \ < 2 e^{-\nu L_n} \ ,
  \end{equation}
  whenever $\nu > \frac{5\alpha}{8}\mu$. 2) Using eqs. \eqref{eq:lambdaj}
  and \eqref{eq:scalesum2}, we see that, for $j=1,\ldots,m_E$,
  \begin{equation}\label{eq:lambdajinI}
    |\lambda_E^j-E| \ \le \ \sum_{k=1}^j |\lambda_E^k-\lambda_E^{k-1}| \ \le \ 4\sqrt{d}  e^{-\frac{7\alpha}{8} \mu L_j} \ < \ r \ ,
  \end{equation}
  since
  $L_{j}\ge L_{\mathrm{fin}}>\frac{8}{7\mu} \log
  \tfrac{4\sqrt{d}}{r}$.  Thus each eigenvalue is in the localization
  regime, $\lambda_E^j\in I_{\mathrm{AL}}$.
\end{Rem}
\begin{proof} Let $\Lambda=\Lambda_L(x_0)$ and
  $\Sigma = \Sigma(\Lambda_L(x_0)).$ By (A3), with probability at
  least $1-A_{\mathrm{AL}}^2 e^{\frac{\alpha \mu}{4}}L^{-p}$, we have
  \begin{equation}\label{eq:initial}
    \max_{E\in \Sigma} \max_{y\in \mathcal{C}(\varphi_E)\cap \Lambda} \left ( \sum_{x\in \Omega} |\varphi_E(x)|^2 e^{2\mu |x-y|} \right )^{\frac{1}{2}} \ \le \ L^{\frac{d+p}{2}} e^{-\frac{\alpha}{8}\mu} \ \le \ e^{\frac{\alpha}{8}\mu L_1} .
  \end{equation}
  On the event that eq.\ \eqref{eq:initial} holds, for each $E\in \Sigma$
  we will construct sequences $(Q_E^j)_{j=1}^{m_E}$
  ,$(\psi_E^j)_{j=0}^{m_E}$, $(\lambda_E^j)_{j=0}^{m_E}$ satisfying
  (1)-(4) as well as localization centers
  $x_E^j\in \mathcal{C}(\psi_E^j)$, for $j=0,\ldots,m_E$, such that
  $x_E^j\in Q_E^{j+1}$ and
  \begin{equation}\label{eq:xj}
    \mathcal{C}(\psi_E^{j+1}) \ \subset\ \left \{x : |x-x_E^{j}|<  \tfrac{\alpha}{4}L_{j+1} \right \} \ ,
  \end{equation}
  for $j=0,\ldots,m_E-1$.
	
  Fix $E\in \Sigma$ and let $\psi_E^0=\varphi_E$, $\lambda_E^0=E$.  For
  ease of notation, we take $Q_E^0=\widetilde{Q}_E^0=\Omega$. We
  define the remainder of the sequence recursively. Let $n\ge 0$ and
  suppose we have already found $(Q_E^j)_{j=0}^{n}$
  ,$(\psi_E^j)_{j=0}^{n}$, $(\lambda_E^j)_{j=0}^n$ and
  $(x_E^j)_{j=0}^{n-1}$ with the desired properties. We note that
  \begin{equation}\label{eq:recursiveM}
    M_n \ := \ \left ( \sum_{x\in \widetilde{Q}_E^{n}} |\psi_E^n(x)|^2 e^{2\mu |x-x_E^n|} \right )^{\frac{1}{2}} \ \le \  e^{\frac{\alpha}{8}\mu L_{n+1} } \ ;
  \end{equation}
  for $n=0$ this follows from \eqref{eq:initial}, while, for $n\ge 1$,
  this holds since $Q_E^n$ is $\epsilon$-good and
  $\epsilon \le \frac{\alpha}{8}\mu$.  If $n=n_{\mathrm{fin}}$ or every
  cube $Q\in \Gc^{n+1}$ with $\mathcal{C}(\psi_E^n)\cap Q \neq \emptyset$ is
  $\epsilon$-bad, then we choose $x_E^n$ to be an arbitrary element of
  $\mathcal{C}(\psi_E^n)$, set $m_E=n$, and there is nothing further
  to show.  Otherwise, pick an $\epsilon$-good cube
  $Q_E^{n+1}\in \Gc^{n+1}$ with
  $Q_E^{n+1} \cap \mathcal{C}(\psi_E^{n})\neq \emptyset$ and pick
  $x_E^n\in Q_E^{n+1} \cap \mathcal{C}(\psi_E^{n})$.  Since $Q_E^{n+1}$ is
  $\epsilon$-good, eq.\ \eqref{eq:minamilemma} holds with
  $\delta=e^{-\epsilon L_{n+1}}$.  Furthermore, eq.\ \eqref{eq:Mbound}
  follows from eqs.\ \eqref{eq:recursiveM}, since
	$$ M_n \ \le \ e^{\frac{\alpha}{8}\mu L_{n+1} } \ \le \ \frac{e^{\frac{5\alpha}{8} \mu L_{n+1}}}{8\sqrt{d}( \frac{5\alpha}{2} L_{n+1})^{\nicefrac{d}{2}} } \delta  e^{\frac{3\alpha}{8}\mu L_{n+1}} \  \le \ \frac{\delta e^{\alpha \mu L_{n+1}}}{8\sqrt{5}} \ < \ \frac{\delta e^{\alpha \mu L_{n+1}}}{2\sqrt{d}} \ , $$
	where we have used \eqref{eq:Lfinbounds} and the fact that
        $\alpha \ge \frac{1}{2}$.  Hence by Lemma \ref{le:1}, there is an
        eigenfunction $\psi_{E}^{n+1}$ on $\widetilde{Q}_{E}^{n+1}$
        with eigenvalue $\lambda_E^{n+1}$ such that \eqref{eq:lambdaj}
        and \eqref{eq:psij} hold for $j=n+1$.
	
	It remains to show that
        $\widetilde{Q}_E^{n+1}\subset \widetilde{Q}_E^n$ and that eq.\
        \eqref{eq:xj} holds for $j=n$.  In fact eq. \eqref{eq:xj}
        (with $j=n $) directly implies that
        $\widetilde{Q}_E^{n+1}\subset \widetilde{Q}_E^{n}$.  Indeed, since
        $x_E^{n}\in Q_E^{n+1}$ and $x_E^{n-1}\in Q_E^{n}$, we
        have \begin{multline*} \widetilde{Q}_{E}^{n+1}\ \subset\ \{ x
          : |x-x_E^{n}| \le (1+\alpha)L_{n+1}\} \ \subset \ \{ x :
          |x-x_E^n| \le (\tfrac{1}{4} + \tfrac{\alpha}{4} )L_{n} \} \\
          \subset \ \{ x : |x-x_E^{n-1}| < (\tfrac{1}{4} + \tfrac{\alpha}{2} )
          L_{n} \}\ \subset \ \{ x : |x-x_E^{n-1}| < \alpha L_{n} \} \ \subset \
          \widetilde{Q}_E^n \ ,
	\end{multline*}
	where we have noted that $\alpha \ge \frac{1}{2}$ in the last step.
        To verify eq.\ \eqref{eq:xj} for $j=n$, consider the
        $\ell^2$-norm of $\psi^{n+1}_E$ on the set $S^c$, with
        $S=\{x : |x-x_E^n|< \tfrac{\alpha}{4}L_{n+1} \}$.  By
        \eqref{eq:psij} and \eqref{eq:awayfromcenter}, we have
	\begin{align*}
          \| \car_{S^c} \psi_E^{n+1}\|_{\ell^2(\widetilde{Q}_E^{n+1})} \ \le& \ \| \psi_E^{n+1} - \car_{\widetilde{Q}_E^{n+1}} \psi_E^n\|_{\ell^2(\widetilde{Q}_E^{n+1})} \ + \ \| \car_{S^c} \psi_E^n \|_{\ell^2(\widetilde{Q}_E^{n})} \\
          \le& \ 3\sqrt{d}e^{-(\alpha \mu-3\epsilon)L_{n+1}} \ + \ e^{-(\frac{\alpha}{4}\mu-\epsilon)L_{n+1}} \\ \le& \ (3\sqrt{d} e^{-\frac{\alpha}{2}\mu L_{\mathrm{fin}}}  +1 ) e^{-\frac{\alpha}{8}\mu L_{n+1}} \ .
	\end{align*}
	Thus, by \eqref{eq:Lfinbounds} and the fact that $\alpha \ge
        \frac{1}{2}$,
	\begin{equation}\label{eq:l2out}
          \| \car_{S^c} \psi_E^{n+1}\|_{\ell^2(\widetilde{Q}_E^{n+1})}  \ \le \  \frac{3\cdot 2^{-\nicefrac{4}{5}} \sqrt{d}  +1}{8\sqrt{d}} \frac{1}{( \frac{\alpha }{2} L_{n+1} + 1)^{\nicefrac{d}{2}}} \ < \ \frac{3}{8} \frac{1}{\sqrt{\#S}}   \ 
	\end{equation} 
	and
	\begin{equation}\label{eq:l2in}
          \| \car_{\{|x-x_E^n|< \frac{1}{8}L_{n+1}\}} \psi_E^{n+1}\|_{\ell^2(\widetilde{Q}_E^{n+1})} \ > \ 1 - \frac{3}{8} \frac{1}{\sqrt{\#S}} \ \ge \ \frac{5}{8}.
	\end{equation}
	For the $\ell^\infty$ norms, eqs.\ \eqref{eq:l2out} and
        \eqref{eq:l2in} imply
	\begin{equation}\label{eq:linfinity}
          \| \car_{S^c} \psi_E^{n+1}\|_{\ell^\infty(\widetilde{Q}_E^{n+1})} \ < \ \frac{3}{8} \frac{1}{\sqrt{\#S}} \ < \ \frac{5}{8} \frac{1}{\sqrt{\#S}} \ < \  \| \car_{S} \psi_E^{n+1}\|_{\ell^\infty(\widetilde{Q}_E^{n+1})} \ .
	\end{equation}
	In particular
        $\mathcal{C}(\psi_E^{n+1}) \subset S= \{x :|x-x_E^n| <
        \frac{\alpha}{4}L_{n+1} \}.$
	
	To see that maps of the form
        $E \mapsto (Q_{E}^{j_E},\psi_{E}^{j_E},\lambda_E^{j_E})$ are
        one-to-one, we note that, by eqs.\ \eqref{eq:l2out},
        \eqref{eq:psij} and \eqref{eq:scalesum2},
	\begin{align*}
          \|\psi_E^0 - \car_{\widetilde{Q}_{E}^{j_E}}\psi_E^{j_E}\|_{\ell^2(\Omega)} \
          \le& \ 	\sum_{k=0}^{j_E-1} \left (\| \car_{(\widetilde{Q}_E^{k+1})^c} \psi_E^k \|_{\ell^2(\widetilde{Q}_E^k)} \ + \
               \|\car_{(\widetilde{Q}_E^{k+1})}  \psi_E^{k} -   \psi_E^{k+1}  \|_{\ell^2(\widetilde{Q}_E^{k+1})} \right )  
          \\ \le&  \   \  \sum_{k=j}^{m_E-1} \left ( e^{-(\alpha \mu -\epsilon) L_{j+1}} \ + \ 3\sqrt{d} e^{-(\alpha \mu -3\epsilon)L_{k+1}} \right ) \\
          \le &  \ (6 \sqrt{d}+2 ) e^{-\frac{5\alpha }{8} \mu  L_{m_E}} \ \le \ 
		\frac{3\sqrt{d}+1}{4\sqrt{d}(\frac{5\alpha}{2}L_{m_E}+1)^{\nicefrac{d}{2}}} \ \le \ \frac{1}{\sqrt{6}} \ < \ \frac{1}{2} \ ,
	\end{align*}
	where we have used \eqref{eq:Lfinbounds} and the facts that
        $L_{m_E}\ge L_{\mathrm{fin}}\ge 4$ and $\alpha \ge \frac{1}{2}$ in
        the last step.  Since
        $\|\psi_E^0-\psi_{E'}^0\|_{\ell^2(\Omega)} = \sqrt{2} $ for
        distinct eigenvalues $E$ and $E'$, we conclude that each such
        map is one-to-one.
	
	A similar calculation leads to eq.\
        \eqref{eq:centersclose}. Let $y\in \mathcal{C}(\psi_E^{m_E})$ and
        set $T=\{x : |x-y|<\frac{\alpha}{16}L_{m_E} \}$.  By \eqref{eq:xj},
        $|y-x_E^{m_E-1}| <\frac{\alpha}{4} L_{m_E}$ and thus
        $T\ \subset \ \{ x : |x-x_E^{m_E-1}| < \frac{5}{16}L_{m_E} \} \ \subset \ \{
        x : |x-x_E^{m_E-1}| < \frac{1}{2}L_{m_E} \} \ \subset \
        \widetilde{Q}_E^{m_E} $ since $x_E^{m_E-1}\in Q_{E}^{m_E}$.  For
        the $\ell^2$ norm of $\psi_E^j$ on $T^c$, we have
	\begin{align*}
          \| \car_{T^c} \psi_E^j \|_{\ell^2(\widetilde{Q}_E^j)} \ 
          \le& \ 	\sum_{k=j}^{m_E-1} \left (\| \car_{(\widetilde{Q}_E^{k+1})^c} \psi_E^k \|_{\ell^2(\widetilde{Q}_E^k)} \ + \
               \|\car_{(\widetilde{Q}_E^{k+1})}  \psi_E^{k} -   \psi_E^{k+1}  \|_{\ell^2(\widetilde{Q}_E^{k+1})} \right ) \ + \ \| \car_{T^c} \psi_E^{m_E}\|_{\ell^2(\widetilde{Q}_E^{m_E})} 
          \\ \le&  \   (6 \sqrt{d}+2) e^{-\frac{5\alpha }{8} \mu  L_{m_E}} \  + \ e^{\frac{\epsilon}{4} L_{m_E}} e^{-\frac{\alpha}{16}\mu L_{m_E}} \\
          \\ \le& \ \left ( 1 + (6\sqrt{d} +2) e^{-\frac{19\alpha }{32} \mu  L_{m_E}}  \right ) e^{-\frac{\alpha}{32}\mu L_{m_E}} 
	\end{align*}
	by eq.\ \eqref{eq:l2out}, eq.\ \eqref{eq:psij} and
        Prop. \ref{pro:scalesum}.  Thus, by \eqref{eq:Lfinbounds} and
        the fact that $\alpha \ge \frac{1}{2}$,
	\begin{equation*}
          \| \car_{T^c} \psi_E^j \|_{\ell^2(\widetilde{Q}_E^j)} \ \le \ 
          \left (1 +  \frac{3\sqrt{d}+1} {4 \sqrt{d} (\frac{19\alpha}{8}L_{m_E}+1)^{\nicefrac{d}{2}}} \right) \frac{1}{8\sqrt{d}}
          \frac{1}{(\frac{\alpha}{8}L_{m_E}+1)^{\nicefrac{d}{2}}}  \ < \  \frac{1}{4} \frac{1}{\sqrt{\# T}}.
	\end{equation*}
	We conclude that
        $ \|1_{T^c}\psi_E^j\|_{\ell^\infty(\widetilde{Q}_E^j)} <
        \|1_{T}\psi_E^j\|_{\ell^\infty(\widetilde{Q}_E^j)}$. Thus,
        $\mathcal{C}(\psi_E^j)\subset T$, which is eq.\
        \eqref{eq:centersclose}.

        Finally, to prove \eqref{eq:improvedbound}, let
        $0\le j\le m_E$ and let $y_j\in \mathcal{C}(\psi_E^j)$.  Then, by
        eq.\ \eqref{eq:centersclose},
        \begin{align*}
          \left \| e^{\nu|\bullet -y_j|}\psi_E^j \right \|_{\ell^2(\widetilde{Q}_E^k)} \ \le& \ \sum_{k=j}^{m_E-1} \left (\| e^{\nu|\bullet -y_j|}\car_{(\widetilde{Q}_E^{k+1})^c} \psi_E^k \|_{\ell^2(\widetilde{Q}_E^k)} \ + \
                                                                                              \|e^{\nu|\bullet -y_j|}(\car_{(\widetilde{Q}_E^{k+1})}  \psi_E^{k} -   \psi_E^{k+1} ) \|_{\ell^2(\widetilde{Q}_E^{k+1})} \right ) \\& \quad  + \ \| e^{\nu|\bullet -y_j|} \psi_E^{m_E}\|_{\ell^2(\widetilde{Q}_E^{m_E})} \\
          \le & \ e^{\frac{\alpha}{8}\nu L_{m_E}}  \sum_{k=j}^{m_E-1} \left (\| e^{\nu|\bullet -y_k|}\car_{(\widetilde{Q}_E^{k+1})^c} \psi_E^k \|_{\ell^2(\widetilde{Q}_E^k)} \ + \ \|e^{\nu|\bullet -y_k|}(\car_{(\widetilde{Q}_E^{k+1})}  \psi_E^{k} -   \psi_E^{k+1} ) \|_{\ell^2(\widetilde{Q}_E^{k+1})} \right ) \\
                                                                                            &\quad + \  e^{\frac{\alpha}{8}\nu L_{m_E}} \| e^{\nu|\bullet -y_{m_E}|} \psi_E^{m_E}\|_{\ell^2(\widetilde{Q}_E^{m_E})} \ ,
        \end{align*}
        $y_k\in \mathcal{C}(\psi_E^k)$ for $k=j+1,\cdots,m_E$. Since
        $y_k\in Q_{E}^{k+1}$, we have
        \begin{equation*}
          \| e^{\nu|\bullet -y_k|}\car_{(\widetilde{Q}_E^{k+1})^c} \psi_E^k \|_{\ell^2(\widetilde{Q}_E^k)} \ \le \ e^{-(\mu-\nu)\alpha L_{k+1}} M_k \ \le \ e^{-((\mu-\nu)\alpha - \epsilon)L_{k+1} } \ ,
        \end{equation*}
        and
        \begin{multline*}
          \|e^{\nu|\bullet -y_k|}(\car_{(\widetilde{Q}_E^{k+1})}
          \psi_E^{k} - \psi_E^{k+1} )
          \|_{\ell^2(\widetilde{Q}_E^{k+1})} \\ \le \ e^{\nu
            (1+\alpha)L_{k+1}} \|(\car_{(\widetilde{Q}_E^{k+1})} \psi_E^{k}
          - \psi_E^{k+1} ) \|_{\ell^2(\widetilde{Q}_E^{k+1})} \ \le \
          3\sqrt{d} e^{-(\mu\alpha - \nu (1+\alpha)-2\epsilon)L_{k+1}} \ ,
        \end{multline*}
        by \eqref{eq:psij}.  Also
        $\| e^{\nu|\bullet -y_{m_E}|}
        \psi_E^{m_E}\|_{\ell^2(\widetilde{Q}_E^{m_E})} \le \|
        e^{\mu|\bullet -y_{m_E}|}
        \psi_E^{m_E}\|_{\ell^2(\widetilde{Q}_E^{m_E})} \le
        e^{\frac{\epsilon}{4}L_{m_E}} $. It follows that
        \begin{align*}
          \left \| e^{\nu|\bullet -y_j|}\psi_E^j \right \|_{\ell^2(\widetilde{Q}_E^k)} \ \le& \ 
                                                                                              e^{\frac{\alpha}{8}\nu L_{m_E}}  \left ( \sum_{k=j}^{m_E-1} \left (e^{-((\mu-\nu)\alpha - \epsilon) L_{k+1}}  \ + \ 3\sqrt{d} e^{-(\mu \alpha -\nu(1+\alpha) -3\epsilon)L_{k+1}}  \right ) + e^{\frac{\epsilon}{4} L_{m_E}} \right )  \\
          \le & \ e^{\frac{\alpha}{8}\nu L_{m_E}} \left ( \frac{1+3\sqrt{d}}{e^{(\mu \alpha -\nu(1+\alpha) -3\epsilon)L_{m_E}}-1} + e^{\frac{\epsilon}{4} L_{m_E}} \right ) \\
          \le & \ \left (1 +  \frac{1+3\sqrt{d}}{e^{\epsilon L_{m_E}}-1}  \right )e^{(\frac{\alpha}{8}+ \frac{\epsilon}{4})\nu L_{m_E}}  \ ,
        \end{align*}
        where we have used Prop.\ \ref{pro:scalesum} and the facts
        that $(1+\alpha)\nu <\alpha \mu - 3\epsilon$.
      \end{proof}

      Lemma \ref{le:iteration} establishes an improved bound on
      eigenfunctions for which the iteration proceeds to scale
      $L_{m_E}$ with $m_E\ge 2$.  To prove Theorem \ref{thr:1} we will
      estimate the number of eigenfunctions for which such an
      improvement is possible.  This will be accomplished by using
      large deviation estimates to bound the number of bad boxes of a
      given generation.  To start we need a bound on the probability
      that a box of generation $n$ is bad.

      In the arguments below, we fix parameters $\epsilon$, $\alpha$ and
      $\beta$, as above. The symbol $c$ will be used for unspecified
      constants, depending on $\alpha$, $\beta$, $\epsilon$ and the various
      parameters appearing in (A1)-(A4), but independent of $L$ and
      the generation $n$.  The notation $A \lesssim B$ (resp.
      $A \gtrsim B$) indicates $A\le c B$ (resp. $A\ge cB$).
 
      \begin{Pro}\label{prop:badcube}For $Q\in \Gc^n$, we have
	\begin{equation}\label{eq:epsilonbad}
          \Pr (\text{$Q$ is $\epsilon$-bad}) \ \lesssim \ e^{-\frac{\epsilon}{3} L_n} \ .
	\end{equation}
      \end{Pro}
      \begin{proof}
	By (A3), \begin{equation}\label{eq:badloc}
          \Pr(M_\omega(Q)>e^{\epsilon \lfloor \frac{1}{4} L_n
            \rfloor}) \ \le \ A_{AL}^2 \# Q e^{-\frac{\epsilon}{2}L_n}
          \ \le \ A_{AL}^2 L_n^d e^{-\frac{\epsilon}{2}L_n} \ .
	\end{equation}
        By (A4), we have
	$$ \Pr \left ( \tr (\car_{E +[-2e^{-\frac{\epsilon}{2} L_n },2e^{-\frac{\epsilon}{2}L_n}]}((H_\omega)_Q)) \ge 2 \right ) \ \le \ 4 A_M (\#Q)^2 e^{-\epsilon L_n} \ \le \ 4 A_M L_n^{2d}e^{-\epsilon L_n} \ ,$$
        for any $E\in I_{\mathrm{AL}}$.  Since $I_{\mathrm{AL}}$ is a
        finite union of intervals, we can find
        $m\le c e^{\frac{\epsilon}{2} L_n} $ points
        $E_1,\ldots,E_m\in I_{\mathrm{AL}}$ such that for any
        $E\in I_{\mathrm{AL}}$ we have
        $|E-E_j| \le e^{-\frac{\epsilon}{2}L_n}$ for some
        $j=1,\ldots,m$, and thus
$$\tr (\car_{E_j +[-2e^{-\frac{\epsilon}{2} L_n },2e^{-\frac{\epsilon}{2}L_n}]}((H_\omega)_Q)) \ \ge \ \tr (\car_{E +[-e^{-\frac{\epsilon}{2} L_n },e^{-\frac{\epsilon}{2}L_n}]}((H_\omega)_Q)) \ . $$
Therefore
\begin{multline}\label{eq:badminami}
  \Pr \left ( \text{for some } E\in I_{\mathrm{AL}}\ , \quad \tr (\car_{E
      +[-e^{-\frac{\epsilon}{2} L_n
      },e^{-\frac{\epsilon}{2}L_n}]}((H_\omega)_Q)) \ge 2\right ) \\
  \le \ \ \sum_{j=1}^m \Pr \left ( \tr (\car_{E_j
      +[-2e^{-\frac{\epsilon}{2} L_n
      },2e^{-\frac{\epsilon}{2}L_n}]}((H_\omega)_Q)) \ge 2 \right ) \
  \le \ m \cdot c L_n^{2d}e^{-\epsilon L_n} \ \le \ c
  L_n^{2d}e^{-\frac{\epsilon}{2} L_n} \ .
\end{multline}
Eq.\ \eqref{eq:epsilonbad} follows from eqs.\ \eqref{eq:badloc} and
\eqref{eq:badminami}.
\end{proof}

For a cube $Q\in \Gc^n$, with $n\ge 1$, let $\mathrm{B}(Q) $ denote the
event that $Q$ is $\epsilon$-bad. The event $\mathrm{B}(Q) $ depends
only on the realization of the random potential in the cube
$\widetilde{Q}$.  Two such cubes $\widetilde{Q}_{\mb{k}_1}^n$ and
$\widetilde{Q}_{\mb{k}_2}^n$ are non-overlapping whenever
$|\mb{k}_1 -\mb{k}_2|\ge 2\alpha +1 $.  It follows that, for each
$\mb{j}\in \{0,1,\ldots,2\alpha\}^d$, the events
$(\mathrm{B}(Q_{\mb{k}}^n))_{\mb{j}+\mb{k} \in (2\alpha +1)\Z^d}$ are mutually
independent.  By a simple extension of standard large deviation
estimates for independent random variables (see Prop.\
\ref{prop:LDP}), we have the following
\begin{Le}
  \label{le:2}
  Let $0<\epsilon <\frac{2d}{\beta}$. Then, there are $\gamma >0$ and
  $L_{\mathrm{fin}}$ sufficiently large so that, for
  $L\geq \exp(\nicefrac{L_{\mathrm{fin}}}{\beta})$, if
  $\Lambda_L(x_0)\subset \Omega$, then
  \begin{multline}
    \label{eq:6}
    \pro\left(\text{For each $n=1,\ldots,n_{\mathrm{fin}}(L)$,
      }\#\{Q\in \Gc^n(\Lambda_L(x_0)) : \ Q \text{ is $\epsilon$-good} \}
      \ \geq \ \# \Gc^n(\Lambda_L(x_0)) (1-
      e^{-\frac{\epsilon}{4}L_n}) \right)\\ \geq \ 1- e^{-L^{\gamma}} \
  \end{multline}
\end{Le}
\begin{proof}
  By taking $L_{\mathrm{fin}}$ large enough, we have, by Proposition
  \ref{prop:badcube},
  \begin{equation*}
    \Pr(\mathrm{B}(Q)) \ \le \   \tfrac{1}{2}
    e^{-\frac{\epsilon}{4}L_n} \ ,
  \end{equation*}
  for $Q\in \Gc^n$, $n=1,\ldots,n_{\mathrm{fin}}$.  By Prop.\
  \ref{prop:LDP}, for any $\delta\in [0,1]$,
  \begin{equation*}
    \pro \left ( \# \left \{ Q\in \Gc^n(\Lambda_L(x_0)) \ : \
        \text{$Q$ is $\epsilon$-bad} \right  \} \ge   \#
      \Gc^n(\Lambda_L(x_0))  \left (\delta +  \pro \left (
          \mathrm{B}(Q) \right ) \right ) \right ) \ \le \ \exp \left
      (-  \frac{\delta^2 }{3 (2\alpha +1)^d}   \#
      \Gc^n(\Lambda_L(x_0))  \right ) .
  \end{equation*}
  Taking $\delta = \frac{1}{2} e^{-\frac{\epsilon}{4} L_n}$ and
  $\mathrm{G}_n = \left \{ \#\{Q\in \Gc^n(\Lambda_L(x_0)) : \ Q \text{ is
      $\epsilon$-good} \} \geq \# \Gc^n(\Lambda_L(x_0)) (1-
    e^{-\frac{\epsilon}{4}L_n}) \right \}$, we have, for
  $L_{\mathrm{fin}}$ large enough,
  \begin{equation}
    \label{eq:ALQldp}
    \pro(\mathrm{G}_n) \ \ge \ 1- \exp\left (-c
      L_n^{-d}e^{-\frac{\epsilon}{2}L_n} L^d  \right ) \ \ge \  1-
    \exp\left (- c\left(  \nicefrac{L^{d-\frac{\beta\epsilon}{2}}}{(\log
        L)^d}  \right )  \right ) \ ,
  \end{equation}
  where we have used \eqref{eq:Njnbound} and the bound
  $L_n\le L_1 \le \beta \log L $.  Note that the event whose probability
  is estimated in \eqref{eq:6} is
  $\mathrm{G}_{n_{\mathrm{fin}}} \cap \cdots \cap \mathrm{G}_1 $.  Using eq.\
  \eqref{eq:ALQldp} for each $n$, we see that
  \begin{equation*}
    \Pr(  \mathrm{G}_{n_{\mathrm{fin}}} \cap \cdots \cap \mathrm{G}_1 ) \ \ge
    \ 1 - \sum_{n=1}^{n_{\mathrm{fin}}} (1-\Pr(\mathrm{G}_{n})) \ \ge
    \ 1 - n_{\mathrm{fin}} \exp\left (- c\left(
        \nicefrac{L^{d-\frac{\beta\epsilon}{2}}}{(\log L)^d}  \right
      )\right ) \ . 
  \end{equation*}
  Since $n_{\mathrm{fin}}\lesssim \log \log L$, by \eqref{eq:5}, and
  $\epsilon< \nicefrac{2d}{\beta}$, it follows that eq.\ \eqref{eq:6}
  holds with $\gamma < d-\frac{\beta\epsilon}{2}$ provided $L_{\mathrm{fin}}$
  is large enough.
\end{proof}
We are now ready to prove Theorem \ref{thr:1}.  Given $\Omega$ and
$\Lambda:=\Lambda_L(x_0)\subset \Omega$, consider the event
\begin{equation}
  \label{eq:21}
  \mathrm{G}_{\Lambda} \ = \ \{\text{conclusions of Lemma \ref{le:iteration} hold}\} \cap \mathrm{G}_{n_{\mathrm{fin}}}\cap \cdots \cap \mathrm{G}_1 \ ,  
\end{equation}
where $\mathrm{G}_n$, $n=1,\ldots, n_{\mathrm{fin}}$ are as in the
proof of Lemma \ref{le:2}.  By Lemmas \ref{le:iteration} and
\ref{le:2}, we have
$$\Pr(\mathrm{G}_{\Lambda}) \ \ge \ 1 - A_{\mathrm{AL}}^2 L^{-p}- e^{-L^\gamma} \ \ge \ 1 - c L^{-p}.$$
For the remainder of the proof, we assume that this event occurs.

Let
$\Sigma = \{ E\in [a,b]\cap \Ec((H_\omega)_\Omega) :
\mathcal{C}(\varphi_E)\cap \Lambda \neq \emptyset \}.$ By Lemma
\ref{le:iteration}, there is a one-to-one map
$E \mapsto (Q_E^{m_E},\psi_E^{m_E})$, for $E\in \Sigma$, such that $Q_E^{m_E}$
is a good-cube, $\mathcal{C}(\varphi_E)\subset \widetilde{Q}_E^{m_E}$ and the
inequality \eqref{eq:improvedbound} holds.  From
\eqref{eq:improvedbound}, we see that
\begin{equation*}
  \left (\sum_{\Omega} e^{2\nu(|x-y|-\ell)_+} |\varphi_E(x)|^2 \right )^{\frac{1}{2}} \ \le \ 1 + e^{-\nu \ell} \left (1 +  \frac{1+3\sqrt{d}}{e^{\epsilon L_{m_E}}-1}  \right ) e^{(\frac{\alpha}{8}+ \frac{\epsilon}{4})\nu L_{m_E}} \ .
\end{equation*}
It follows that
\begin{equation}
  \label{eq:ellnu}
  \ell_\nu(\varphi_E) \le (\frac{\alpha}{8}+ \frac{\epsilon}{4}) L_{m_E} + \log  \left (1 +  \frac{1+3\sqrt{d}}{e^{\epsilon L_{\mathrm{fin}}}-1}  \right ) \ \le \ \frac{\alpha + 3 \epsilon}{8} L_{m_E} \ ,
\end{equation}
for $L_{\mathrm{fin}}$ large enough. Thus
\begin{equation*}
  \{ E\in \Sigma \ :\ \ell_\nu(\varphi_E) > \ell \} \ \subset\ \left \{
    E \in \Sigma : L_{m_E} > \frac{8}{\alpha + 3\epsilon}\ell \right \} \ ,
\end{equation*}
for $\ell\ge \frac{\alpha + 3\epsilon}{8}L_{\mathrm{fin}}$.

Consider now the case that $m_E=n <n_{\mathrm{fin}}$.  If this holds,
then by Lemma \ref{le:iteration}, every cube $Q\in \Gc^{n+1}(\Lambda)$
such that $Q\cap \mathcal{C}(\psi_E^{m_E})\neq \emptyset$ and
$\widetilde{Q}\subset \widetilde{Q}_E^{m_E}$ is $\epsilon$-bad . Pick one
such cube, $Q$. From eq.\ \eqref{eq:centersclose}, it follows that
$\mathcal{C}(\varphi_E) \subset\widetilde{Q}$.  Thus, we have shown
that
$$\{E\in \Sigma \ : \ m_E=n \} \ \subset \ \{ E \in \Sigma \ :\ \mathcal{C}(\varphi_E) \subset \widetilde{Q} \subset \widetilde{Q}_E^n \text{ with $Q\in \Gc^{n+1}$ an $\epsilon$-bad cube.}\}.$$
For each $E$, let $n_E$ be the smallest integer $n$ such that
$\mathcal{C}(\varphi_E) \subset \widetilde{Q} \subset \widetilde{Q}_E^n$ with
$Q\in \Gc^{n+1}$ an $\epsilon$-bad cube.  Thus
\begin{equation*}
   \{ E\in \Sigma \ :\  \ell_\nu(\varphi_E) > \ell \} \ \subset \ \bigcup_{n \ : \ L_n > \frac{8}{\alpha + 3\epsilon}\ell}  \{ E \in \Sigma \ :\ n_E = n \}.
\end{equation*}
Note that $0\le n_E\le m_E$.  Thus, by Lemma \ref{le:iteration}, the
map $E \mapsto (Q_E^{n_E},\psi_E^{n_E})$ is one-to-one. On the event
$\mathrm{G}_\Lambda$, the number of bad cubes of generation $(n+1)$ is
bounded by
$\#\Gc^{n+1}(\Lambda_L(x_0))e^{-\frac{\epsilon}{4}L_{n+1}}$.  For each
such cube, there are at most $(1+2\alpha)^d$ cubes $Q'\in \Gc^n$ such that
$\widetilde{Q}'\supset \widetilde{Q}$. Thus
\begin{equation*}
  \#\{ E \in \Sigma \ :\ n_E=n\} \ \le \ \#\Gc^{n+1}(\Lambda_L(x_0))e^{-\frac{\epsilon}{4}L_{n+1}} \times (1+2\alpha)^d \times L_n^d \ ,
\end{equation*}
where $L_n^d$ is the number of eigenvalues for the Hamiltonian
$(H_\omega)_{Q}$ restricted to a cube of generation $n$.  By
\eqref{eq:Njnbound}, we see that
$\#\{ E \in \Sigma \ :\ n_E=n\} \ \lesssim \ L^d
e^{-\frac{\epsilon}{4}L_{n+1}} \ .$ Thus
\begin{equation*}
  \begin{split}
    \# \{ E\in \Sigma \ :\ \ell_\nu(\varphi_E) > \ell \} \ &\lesssim \
    \sum_{ L_n > \frac{8}{\alpha + 3\epsilon}\ell } L^d
    e^{-\frac{\epsilon}{4}L_{n+1}} \ \lesssim \ L^d \sum_{ L_n >
      \frac{8}{\alpha + 3\epsilon}\ell } L^d e^{-\frac{\epsilon}{16}L_{n}} \\
    &\lesssim \ \frac{1}{\exp(\frac{1}{2\alpha +6\epsilon}\ell  ) -1} L^d \
    \lesssim \ L^d e^{- \frac{1}{2} \frac{1}{\alpha +3\epsilon }\ell }\,
  \end{split}
\end{equation*}
by Proposition \ref{pro:scalesum}, provided
$t\ge \frac{\alpha + 3\epsilon}{8}L_{\mathrm{fin}}$. Taking into account
the restrictions $(1+\alpha)\nu\leq\alpha \mu +3\varepsilon$ and
$\alpha\geq1$, as $\varepsilon$ can be chosen arbitrarily small, we
can pick it so that $\frac{1}{2} \frac{1}{\alpha +3\epsilon }\geq C_\nu$
where $C_\nu$ is defined in Theorem~\ref{thr:1}. This completes the
proof of Theorem~\ref{thr:1}.
\subsection{Sketch of the proof of
  Theorem~\ref{thr:8}}\label{sec:thr8}
Let us now describe the modifications needed to derive
Theorem~\ref{thr:8} for the more general model.

The first set of modifications comes from the fact that we re dealing
with PDEs rather than finite difference equations. In
Lemmas~\ref{le:1} and~\ref{le:iteration}, we use smooth cut-offs and
elliptic regularity to carry over the known decay for the
eigenfunctions to their gradient. Of course, the sub-exponential decay
also worsens the estimate a bit but not in a crucial way. Finally, we
have only independence at a distance. So to obtain the probability
estimate~\eqref{eq:6} that is based on independence, we split our
family of cubes at each generation into $2^d$ families of cubes such
that the members of each family are independent. This works as long as
$L_{n_{\mathrm{fin}}}$ is larger than $r$ (from (IAD)).
  
The second difference comes from the fact that we replaced the Minami
estimate by the spacing estimate (SE). In the proofs of
Proposition~\ref{prop:badcube} and, thus, Lemma~\ref{le:2}, this
worsens a bit the estimate of the probability of~\eqref{eq:36} being
satisfied (at generation $n$): one obtains that this probability is
now larger than $1-L_n^{2d}(\log L_{n-1})^{K}=1-C^KL_n^{2d-K}$ where
$K>0$ is arbitrary; choosing $K$ sufficiently large, the lower bound
in~\eqref{eq:6} now becomes $1-CL_n^{-p}$; we, thus, recover the
conclusion of~\eqref{eq:21}.
  
Finally, one can notice an additional $\log L$ factor in the
probability of bad events in Theorem~\ref{thr:8} (when compared to
Theorem~\ref{thr:1}). This additional factor is obtained to pass from
the estimate on the number of eigenfunction of a certain sup norm for
fixed $t$ to that for arbitrary $t$ (see~\eqref{eq:10}); in the case
of Theorem~\ref{thr:1}, $p$ can be taken arbitrary; in
Theorem~\ref{thr:8}, it is fixed given by the assumption (Loc).

\section{The proof of Theorem~\ref{thr:3}}
\label{sec:proof-theorem-3}
%

One easily relates the onset length of a eigenvector to its sup norm
and proves
\begin{Le}
  \label{le:3}
  If $\|\varphi\|_{\ell^2(\Omega)}=1$ and $M_\ell^\mu (\varphi;y)\leq2$,
  for $\mu >0$ and $\ell\geq0$, one has
  \begin{equation}
    \label{eq:20} \frac1{\sqrt{2}}
    \frac1{(2\ell+2\kappa+1)^{\nicefrac{d}{2}}} \ \leq \
    \|\varphi\|_\infty \ := \ \sup_{x\in\Omega}|\varphi(x)|.
  \end{equation}
  where $\kappa>0$ is such that $8 e^{-2\mu\kappa}\leq1$
\end{Le}
\begin{proof}
  As $\|\varphi\|_{\ell^2(\Omega)}=1$ and $M_\ell^\mu (\varphi;y)\leq2$,
  one computes
  \begin{equation*}
    \begin{split}
      1&=\sum_{x\in\Omega}|\varphi(x)|^2\leq
      \|\varphi\|^2_\infty\sum_{\substack{x\in\Omega
          \\|x-y|\leq\ell+\kappa}}1
      +\sum_{\substack{x\in\Omega\\\ell+\kappa<|x-y|}}
      e^{-2\mu(|x-y|-\ell)_+} e^{2\mu(|x-y|-\ell)_+}
      |\varphi(x)|^2\\&\leq (2\ell+2\kappa+1)^d\|\varphi\|^2_\infty+4
      e^{-2\mu\kappa}.
    \end{split}
  \end{equation*}
  Thus, one has $(2\ell+2\kappa+1)^{-d}\leq 2\|\varphi\|^2_\infty$, that
  is,~\eqref{eq:20}.
\end{proof}
%
For localized eigenfunctions, Lemma~\ref{le:3} provides a lower bound
on the onset length in terms of the sup norm of the
eigenfunction. Notice that there does not exist a reverse bound: the
onset length of an eigenfunction may be large even though its sup norm
is of order 1. Indeed, think of the two lowest eigenfunctions of a
symmetric double well that is widely spaced.\\

One easily relates the sup norm of an eigenvector to a bound on its
gradient and proves
\begin{Le}
  \label{le:4}
  Pick $\Omega=\Z^d$. For $\varphi\in\ell^2(\Z^d)$, one has
  \begin{equation}
    \label{eq:7}
    \|\varphi\|_\infty\leq
    4d\|\nabla\varphi\|^{\frac{d}{d+1}}_2\|\varphi\|^{\frac1{d+1}}_2
  \end{equation}
  where
  \begin{equation*}
    \|\nabla
    \varphi\|_2^2=\sum_{x\in\Lambda}\sum_{|e|_1=1}|\varphi(x+e)-
    \varphi(x)|^2. 
  \end{equation*}
\end{Le}
\begin{proof}
  Pick $x_0\in\Lambda$ such that
  $|\varphi(x_0)|=\|\varphi\|_\infty$. Thus, for $v\in\Z^d$, one can
  write $\D x_0+v=x_0+\sum_{k=1}^{|v|_1}(x_{k}-x_{k-1})$ where
  $|x_{k}-x_{k-1}|_1=1$ and $x_i\not=x_j$ if $i\not=j$. Thus, one has
  \begin{equation*}
    \varphi(x_0+v)=\varphi(x_0)+\sum_{k=1}^{|v|_1}(\varphi(x_{k})-\varphi(x_{k-1})).
  \end{equation*}
  Using Cauchy-Schwartz, this yields
  \begin{equation}
    \label{eq:12}
    |\varphi(x_0+v)|\geq\|\varphi\|_\infty
    -\sqrt{|v|_1}\|\nabla\varphi\|_2. 
  \end{equation}
  Either one has $\|\varphi\|_\infty\leq 2\|\nabla\varphi\|_2$; then, as
  $\|\nabla\varphi\|_2\leq 2\sqrt{d}\|\varphi\|_2$, one has
  $\|\varphi\|_\infty\leq
  4d\|\nabla\varphi\|^{\frac{d}{d+1}}_2\|\varphi\|^{\frac1{d+1}}_2$.
  Or one has $\|\varphi\|_\infty\geq 2\|\nabla\varphi\|_2$, hence,
  by~\eqref{eq:12}, for any
  $\D |v|_1\leq\left(\frac{\|\varphi\|_\infty}
    {\|\nabla\varphi\|_2}\right)^2$, one has
  $2|\varphi(x_0+v)|\geq\|\varphi\|_\infty$. This implies
  \begin{equation*}
    4\|\varphi\|_2^2=\sum_{v\in\Z^d} |2\varphi(x_0+v)|^2\geq
    \sum_{|v|_1\leq\left(\frac{\|\varphi\|_\infty}
        {\|\nabla\varphi\|_2}\right)^2} \|\varphi\|^2_\infty\geq \frac1{d!}
    \|\varphi\|_\infty^2
    \left(\frac{\|\varphi\|_\infty}
      {\|\nabla\varphi\|_2}\right)^{2d}
  \end{equation*}
  Thus, one has
  $\D\|\varphi\|_\infty\leq
  \sqrt[d+1]{2\sqrt{d!}}\|\nabla\varphi\|^{\frac{d}{d+1}}_2\|\varphi\|^{\frac1{d+1}}_2\leq
  4d\|\nabla\varphi\|^{\frac{d}{d+1}}_2\|\varphi\|^{\frac1{d+1}}_2$.
\end{proof}
%
Let us complete the proof of Theorem~\ref{thr:3}. Pick $c>0$.  It is
well known that for our choice of $-\Delta$, the infimum of the almost
sure spectrum $E_-$ is given by $E_-=-2d+\essinf\omega_0$ (where
$(\omega_x)_{x\in\Z^d}$ is the random potential. Thus, if
$\varphi_E\in\ell^2(\Z^d)$ is a normalized eigenfunction associated to
an energy $E$ less than $E_-+c\ell^{-d-1}$, one has
\begin{equation*}
  \|\nabla \varphi_E\|^2\leq \|\nabla
  \varphi_E\|^2+\sum_{x\in\Z^d}(\omega_x-\essinf\omega_x)
  |\varphi_E(x)|^2=\langle(H_\omega-E_-)\varphi_E,\varphi_E\rangle
  \leq c\ell^{-d-1}.
\end{equation*}
Applying first Lemma~\ref{le:3} and then Lemma~\ref{le:4}, we get that
\begin{equation*}
  2\ell_\nu(\varphi_E,x_E)+2\kappa+1\geq
  \|\varphi\|^{-2/d}_\infty\geq
  (4d)^{-\frac2d}\|\nabla\varphi_E\|^{-\frac{2}{d+1}}_2
  \geq   (4d)^{-\frac2d}c^{-\frac{2}{d+1}}\,\ell.
\end{equation*}
Thus, if $\ell\geq\max(\kappa,1)$, picking $c>0$ such that
$(4d)^{-\frac2d}c^{-\frac{2}{d+1}}=5$, we get~\eqref{eq:4} and
complete the proof of Theorem~\ref{thr:3}.

\subsection*{Acknowledgments} This material is based upon work
supported by the National Science Foundation under Grant No.\ 1900015
(JS) and in part through computational resources and services
provided by the Institute for Cyber-Enabled Research at Michigan State
University. The authors are grateful to the Institut Mittag-Leffler in
Djursholm, Sweden, where this work was started as part of the program
Spectral Methods in Mathematical Physics in Spring 2019.

\appendix

\section{SULE bound from Eigenfunction Correlators}\label{sec:SULE}
In the literature, spectral localization is frequently expressed via a
bound
\begin{equation}
  \label{eq:spectralloc}
  \sum_{x} e^{\nu |x-y|} \esp \left (Q_\Omega(I,x,y) \right )\ \le \ A 
\end{equation} 
with constants $A$ and $\nu$ independent of $\Omega$, where
$Q_\Omega(I,x,y)$ is the \emph{eigenfunction correlator of $H_\Omega$
  on $I$} (see \cite[Chapter 7]{MR3364516}).  For a finite region
$\Omega$,
\begin{equation*}
  Q_\Omega(I,x,y) \ = \ \sum_{E\in I\cap\sigma((H_\omega)_\Omega)} |\varphi_E(x)||\varphi_E(y)|,
\end{equation*}
where $\varphi_E$ is the normalized eigenvector corresponding to
eigenvalue $E$.  For the operators considered here, the spectrum is
known to be almost surely simple \cite{MR1462151,MR2203783}; for
operators with degenerate spectrum the term
$|\varphi_E(x)||\varphi_E(y)|$ should be replaced by
$|\langle \delta_x, P_E \delta_y\rangle |$, with $P_E$ the corresponding
eigen-projection. For an infinite region, one may replace this
definition with
$$Q_\Omega(I,x,y) \ = \ \sup_{f} \left | f((H_\omega)_\Omega)(x,y) \right | \ , $$
where the supremum is taken over Borel measurable functions $f$ with
support in $I$ and $|f(x)|\le 1$ everywhere.  \emph{A posteriori}, one
concludes from \eqref{eq:spectralloc} that $(H_\omega)_\Omega$ has
pure point spectrum in $I$ (almost surely), and (since the spectrum is
simple) that
\begin{equation}\label{eq:correlator}
  Q_\Omega(I,x,y) \ = \ \sum_{E\in I\cap\Ec((H_\omega)_\Omega))} |\varphi_E(x)||\varphi_E(y)| \ .
\end{equation}

We now recall the derivation of a SULE estimate of the form (A3) from
spectral localization \eqref{eq:spectralloc}.
\begin{Pro} Let $(H_\omega)_\Omega$ be a random operator on a region
  $\Omega\subset \Z^d$ such that $(H_\omega)_\Omega$ has simple, pure-point
  spectrum in $I$ almost surely and \eqref{eq:spectralloc} holds and
  let $\epsilon >0$.  If $S\subset \Omega$ is a finite set, then, with
  probability greater than $1-\epsilon$, every eigenvector $\varphi_E$
  of $(H_\omega)_\Omega$ with eigenvalue $E\in I$ and
  $\mathcal{C}(\varphi_E)\cap S \neq \emptyset$ satisfies
  \begin{equation}\label{eq:SULEinappendix}
    \left ( \sum_{x\in \Omega} e^{\nu|x-y|} |\varphi_E(x)|^2 \right )^{\frac{1}{2}} \ \le \ A \left ( \frac{\#S}{\epsilon} \right )^{\frac{1}{2}}
  \end{equation}
  for any $y\in \mathcal{C}(\varphi_E)\cap S$. In particular, \eqref{eq:SULE}
  holds with $A_{\mathrm{AL}}=A$ and $\mu=\frac{\nu}{2}$.
\end{Pro}

\begin{proof}
  From \eqref{eq:spectralloc}, it follows that
$$\esp \left (\sum_{y\in S} \sum_{x\in \Omega } Q_\Omega(I,x,y)e^{\nu|x-y|} \right ) \ \le \  A \# S .$$
By Markov's inequality, with probability $\ge 1- \epsilon$, we have
$$\sum_{y\in S} \sum_{x\in \Omega }  Q_\Omega(I,x,y)e^{\nu|x-y|} \ \le \ A \frac{  \# S}{\epsilon} $$
from which we conclude, using \eqref{eq:correlator}, that
$$ \sum_{x\in \Omega} e^{\nu|x-y|} |\varphi_E(x)||\varphi_E(y)| \ \le \ A \frac{  \# S}{\epsilon} $$
for every eigenvalue $E\in \Ec(H_\Omega)$ and each $y\in S$.  If
$\mathcal{C}(\varphi_E)\cap S\neq \emptyset$, then taking
$y \in \mathcal{C}(\varphi_E)\cap S$, we have
$$\sum_{x\in \Omega} e^{\nu|x-y|} |\varphi_E(x)| \|\varphi_E\|_\infty  \ \le \ A  \frac{\# S}{\epsilon 
}.$$ Since $|\varphi_E(x)|\le \|\varphi\|_\infty$ for every $x$, we
conclude that
$$\sum_{x\in \Omega} e^{\nu|x-y|}|\varphi_E(x)|^2 \ \le \ A \frac{\# S}{\epsilon}.$$
Taking the square root yields \eqref{eq:SULEinappendix}.
\end{proof}

\section{A large deviation principle}\label{sec:LDP}
\begin{Pro}\label{prop:LDP}
  Let $X_1,\ldots,X_N$ be identically distributed random variables
  with
	$$\Pr[X_j=1] = p \quad \text{ and } \quad \Pr[X_j=0]=1-p.$$
	Suppose there is a partition of $\{1,\ldots,N\}$ into
        $K$-disjoint subsets $S_1,\ldots,S_K$ such that, for each
        $j=1,\ldots,K$, the variables $(X_m)_{m\in S_j}$ are mutually
        independent.  Then for any $\alpha \ge 1$,
	\begin{equation}\label{eq:LDP}
          \Pr \left [ \sum_{m=1}^N X_m >  N  (p+\delta) \right ]  \ \le \ \exp \left ( -\frac{\delta^2 }{3K} N \right ).
	\end{equation}
      \end{Pro}
      \begin{proof}
        Let $Z(t)= \esp[e^{t\sum_m X_m}].$ By H{\"o}lder's inequality and
        the assumption that $(X_m)_{m\in S_j}$ are mutually independent,
  $$ Z(t) = \esp\left [\prod_{j=1}^K e^{ t\sum_{m\in S_j}X_m} \right ] \ \le \ \prod_{j=1}^K \left (\esp \left [ e^{Kt\sum_{m\in S_j}X_m} \right ]\right )^{\nicefrac{1}{K}}  \ = \ \left (1 + p(e^{Kt}-1) \right )^{\nicefrac{N}{K}}  \ \le \ e^{\nicefrac{Np(e^{Kt}-1)}{K} } .$$
  It follows that
  $$\Pr\left [\sum_{m=1}^N X_m > N  (p+\delta) \right ] \ \le \ Z(t)e^{-N (p+\delta)  t} \ \le \ e^{N\left (p \frac{e^{Kt}-1}{K} -(p+\delta) t\right  )} \ \le \ e^{N \left (  \frac{e^{Kt}-1}{K} - (1+\delta) t \right ) },$$
  where in the last step we have used that $e^{Kt}-1 -Kt\ge 0$.
  Optimizing over $t$ yields
  $$\Pr\left [\sum_{m=1}^N X_m > N  (p+\delta) \right ] \ \le \\ e^{\frac{N}{K}\left ( \delta -(1+\delta) \log(1+\delta) \right  )} .$$
  Finally, eq.\ \eqref{eq:LDP} follows since
  $(1+\delta) \log(1+\delta)-\delta \ge \nicefrac{\delta^2}{3}$ for $0\le \delta \le 1$.
\end{proof}

\def\cprime{$'$} \def\cydot{\leavevmode\raise.4ex\hbox{.}}

\end{document}